\documentclass[maxcitenames=2]{fullpaper}
\pdfoutput=1

\title{ARRIVAL in cube-root subexponential time}
\shorttitle{ARRIVAL in cube-root subexponential time}

\author{
\authorname{Ebrahim Ghorbani}
\affiliation{Hamburg University of Technology, Institute for Algorithms and Complexity, Hamburg, Germany}
\authormail{ebrahim.ghorbani@tuhh.de}
\authorname{Raul Lopes}
\affiliation{Hamburg University of Technology, Institute for Algorithms and Complexity, Hamburg, Germany}
\authormail{raul.lopes@tuhh.de}
\authorname{Matthias Mnich}
\affiliation{Hamburg University of Technology, Institute for Algorithms and Complexity, Hamburg, Germany}
\authormail{matthias.mnich@tuhh.de}
}
\shortauthors{Ghorbani, Lopes, and Mnich}
\date{}

\usepackage{algorithm}
\usepackage[noend]{algpseudocode}
\algrenewcommand\algorithmicrequire{\textbf{Input:}}
\algrenewcommand\algorithmicensure{\textbf{Output:}}
\usepackage{mathrsfs}
\usepackage{xspace}
\usepackage{tikz}
\usetikzlibrary{arrows.meta}
\usetikzlibrary{calc}

\DeclareMathOperator{\dist}{dist}
\DeclareMathOperator{\poly}{poly}
\DeclareMathOperator{\sgn}{sgn}
\DeclareMathOperator{\bbox}{bbox}
\DeclareMathOperator{\dep}{dep}
\DeclareMathOperator{\err}{err}
\newcommand{\rd}{\varepsilon}
\DeclareMathOperator{\surplus}{\Delta}
\newcommand{\N}{\mathbb{N}}
\newcommand{\0}{\mathbf{0}}
\newcommand{\1}{\mathbf{1}}
\newcommand{\e}{\mathbf{e}}
\newcommand{\gmap}{\mathscr{F}}
\newcommand{\halfspace}[2]{\mathcal{H}(#1, #2)}
\newcommand{\garrival}{{\sc Multi-Arrival}\xspace}
\newcommand{\arrival}{{\sc Arrival}\xspace}

\tikzset{vertex/.style={draw,circle,minimum size=5pt,inner sep=0pt}}
\tikzset{blackvertex/.style={draw,circle,minimum size=5pt,inner sep=0pt, fill=black}}
\tikzset{arrow/.style={-{Latex[scale=.75]}, shorten >= 0pt}}

\definecolor{goodred}{HTML}{CC6677}

\begin{document}
\pagenumbering{roman}
\maketitle
\begin{abstract}
  In the \arrival problem, a token travels through a switch digraph: at every vertex, it leaves along one of two out-arcs, and the two are used alternately.
  The question is which of two terminals the token reaches first.
  \arrival lies in $\mathsf{NP} \cap \mathsf{coNP}$, but no polynomial-time algorithm is known.
  The fastest known algorithms turn some vertices into terminals and search for consistent guesses of the numbers of tokens that these vertices send out.
  The first, due to G\"artner et al. [ICALP 2021], searches all vertices at once with an algorithm for fixed points of monotone maps, and the second, due to Haslebacher [ICALP 2025], searches one vertex at a time by binary search.
  Both search for a fixed point of the same map, which returns how many tokens arrive at the guessed vertices.

  We guess the vertices in blocks of $k$, between these two extremes, which needs a search in $k$ dimensions with few evaluations of this map.
  To do so, we build on the connection between \arrival and $\ell_1$-contractions established in the second of these works.
  We search with $\ell_1$-centerpoints of the remaining guesses, which are guaranteed to exist by Haslebacher et al. [FOCS 2025], and which we approximate deterministically.
  In order to prove our results, we show how to find the least fixed points of maps from $\{0, \ldots, N\}^k$ to itself that have unit increments.

  For \arrival on $n$-vertex switch digraphs, this gives run time $2^{\mathcal{O}(n^{1/3} \cdot \log^{2/3} n)}$, improving the previous time bound of $n^{\mathcal{O}(\sqrt{n})}$.
  We work with the more general \garrival problem, with many tokens and several terminals, and solve it in the same time when there are at most $2^n$ tokens.
\end{abstract}
\metadatum{Keywords.}{ARRIVAL, switching flows, fixed points, subexponential-time algorithms, digraphs}
\metadatum{Funding.}{\emph{Raul Lopes}: HIDSS-0002 DASHH (Data Science in Hamburg - Helmholtz Graduate School for the Structure of Matter).}
\thispagestyle{empty}
\clearpage
\pagenumbering{arabic}

\section{Introduction}
\label{sec:introduction}
There are few well-studied problems which are known to lie in the intersection $\mathsf{NP} \cap \mathsf{coNP}$, yet for which no polynomial-time algorithm is known.
One such problem is {\sc Arrival}, a zero-player graph game that was introduced by \fulltextcite{DohrauGKMW2017}.
The game considers a token which travels autonomously through a ``switch digraph'' by always traversing one out-arc at each vertex, after which the vertex ``switches'' to its next out-arc in some fixed ordering.
Throughout the paper, we consider the version of \arrival introduced by \textcite{GartnerHH2021}, which takes as input a directed graph~$G$ with two designated terminal vertices~$t$ and~$t'$.
Each non-terminal vertex has out-degree two, with one ``active'' and one ``inactive'' outgoing arc, and has a directed path to some terminal.
The process begins by placing a token at an origin vertex.
As the token traverses~$G$, at each vertex $v$, it leaves through its active out-arc $a$, and then immediately switches the roles of active and inactive out-arc, so that the next time the token reaches $v$, its other out-arc $a'$ is traversed, upon which the ``first'' out-arc $a$ becomes active again.
The \arrival problem is to determine which terminal, $t$ or $t'$, the token reaches first.
One can show that the token always reaches a terminal, since every non-terminal vertex has a path to some terminal.

The term \emph{switch digraph} originates from the original formulation of {\sc Arrival}~\cite{DohrauGKMW2017}, where the assumption that each vertex has out-degree two allows for an interpretation of the vertices as switches in a railway network traversed by a train (token).
It is straightforward to see that, in this setting, {\sc Arrival} can be solved in roughly $n \cdot 2^n$ steps, where $n$ is the number of switches: there are only $2^n$ switch configurations, and checking all of them determines whether the token reaches $t$ or $t'$ first.
Following the token step by step may indeed take exponentially many steps.
In the instance of \zcref{fig:arrival-bad-instance}, for example, the vertex $u$ is at distance $d \ge 2$ from $t$, and a token with origin~$u$ leaves $u$ exactly $2^{d - 2}$ times before it reaches $t$, since every vertex $v_i$ with $i \le d - 2$ sends the token onwards only at every second visit and back to $u$ at the others.
The main result of \textcite{DohrauGKMW2017} was a proof that  {\sc Arrival} lies in $\mathsf{NP} \cap \mathsf{coNP}$.
This raises the question of whether {\sc Arrival} can be solved in polynomial time.

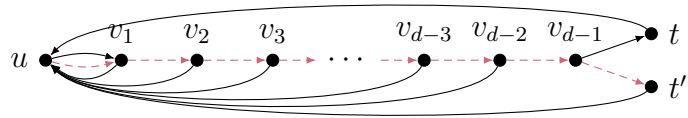
\begin{figure}[ht]
	\centering
	\begin{tikzpicture}
		\node[blackvertex, scale=.89, label=180:$u$] (v0) at (0,0) {};
		\foreach \i in {1,2,3} {
			\node[blackvertex, scale=.9, label=90:$v_\i$] (v\i) at (\i, 0) {};
		}
		\node (dots) at ($(v3) + (1,0)$) {$\cdots$};
		\node[blackvertex, scale=.9, label=90:$v_{d - 3}$] (u1) at ($(dots) + (1, 0)$) {};
		\node[blackvertex, scale=.9, label=90:$v_{d - 2}$] (u2) at ($(dots) + (2, 0)$) {};
		\node[blackvertex, scale=.9, label=90:$v_{d - 1}$] (u3) at ($(dots) + (3, 0)$) {};
		\node[blackvertex, scale=.9, label=0:$t$] (w1) at ($(dots) + (4,0.35)$) {};
		\node[blackvertex, scale=.9, label=0:$t'$] (w2) at ($(dots) + (4,-0.35)$) {};
		\draw[arrow] (v1) .. controls +(-0.3,-0.3) and +(0.4,-0.3) .. (v0);
		\draw[arrow] (v2) .. controls +(-0.3,-0.42) and +(0.5,-0.42) .. (v0);
		\draw[arrow] (v3) .. controls +(-0.3,-0.54) and +(0.6,-0.54) .. (v0);
		\draw[arrow] (u1) .. controls +(-0.3,-0.66) and +(0.7,-0.66) .. (v0);
		\draw[arrow] (u2) .. controls +(-0.3,-0.78) and +(0.8,-0.78) .. (v0);
		\draw[arrow] (w2) .. controls +(-0.5,-0.5) and +(0.9,-0.9) .. (v0);
		\draw[arrow] (w1) .. controls +(-0.5,0.5) and +(0.9,0.9) .. (v0);
		\draw[arrow, goodred,densely dashed] (v0) to [bend right=15] (v1);
		\draw[arrow] (v0) to [bend left=15] (v1);
		\draw[arrow, goodred,densely dashed] (v1) to (v2);
		\draw[arrow, goodred,densely dashed] (v2) to (v3);
		\draw[arrow, goodred,densely dashed] (v3) to (dots);
		\draw[arrow, goodred,densely dashed] (dots) to (u1);
		\draw[arrow, goodred,densely dashed] (u1) to (u2);
		\draw[arrow, goodred,densely dashed] (u2) to (u3);
		\draw[arrow] (u3) to (w1);
		\draw[arrow, goodred,densely dashed] (u3) to (w2);
	\end{tikzpicture}
	\caption{An instance of \arrival with destinations $t$ and $t'$, in which $d \ge 2$ is the distance from~$u$ to $t$. A token with origin $u$ leaves $u$ exactly $2^{d - 2}$ times before it reaches $t$. In \garrival, where $t$ and $t'$ are terminals whose out-arcs lead to $u$, the vertex $u$ is visited $2^{d - 2}$ times by a single token starting at $t$, and $2^{d - 1}$ times if $t$ and $t'$ hold one token each. Black (solid) arcs are the ones used next, and red (dashed) arcs are the other ones.}
	\label{fig:arrival-bad-instance}
\end{figure}

As of now, this question remains unresolved, despite a decade of extensive research~\cite{AugerCD2022,AugerCDE2023,DohrauGKMW2017,FearnleyGMS2021,GartnerHH2021,GartnerHHKMS2018,GhorbaniHM2025,Haslebacher2025,Karthik2017}.
The currently fastest known algorithm for {\sc Arrival} runs in subexponential time $n^{\mathcal{O}(\sqrt{n})}$, and is due to \textcite{GartnerHH2021}
and \textcite{Haslebacher2025}.
For special instances, faster algorithms have been discovered: polynomial-time algorithms exist for tree-like multi-digraphs~\cite{AugerCD2022}, for certain path-like switch digraphs with uniform switch orders~\cite{AugerCDE2023} and for ladders~\cite{GartnerHH2026}.
Quasi-polynomial-time algorithms are known for tree-like multi-digraphs with many tokens~\cite{GhorbaniHM2025}, and for digraphs whose underlying undirected graphs have bounded treewidth~\cite{Haslebacher2025}.

Many of these algorithms are obtained by working with a generalisation of \arrival, known as \garrival, which is due to \textcite{Hoang2022} and is also called {\sc G-Arrival}~\cite{Haslebacher2025}.
In \garrival, multiple tokens (potentially exponentially many in $n$) need to be routed from an arbitrary number of origins to an arbitrary number of terminals (see \zcref{sec:multiarrival} for the formal definition).
Each vertex $v$ has at most two out-arcs, ordered according to a labelling~$\rho$, which determines the switching rule at $v$: whenever a token arrives at $v$, it traverses the next out-arc in the cyclic order specified by $\rho$.
Initially, a total of~$t^+$ tokens are placed on the origin vertices.
Each token then traverses the switch digraph according to the common switching rules until it reaches a terminal, where it remains permanently.
The goal of \garrival\ is to determine the final distribution of tokens among the terminals.
The original \arrival\ problem corresponds to the specific case of a single token and switch digraphs with out-degree two.

Routing the tokens solves every instance of \garrival, but, as for \arrival, it may take exponential time (\zcref{fig:arrival-bad-instance}).
On the other hand, since a vertex sends its tokens along its out-arcs in turns, roughly a $2^{-d}$ fraction of the tokens that it sends follows a shortest path to a terminal, where at most $t^+$ tokens can arrive.
Hence, a vertex within distance~$d$ from a terminal sends at most about $2^d \cdot t^+$ tokens, which explains why instances in which every vertex is within distance $d$ of a terminal can be solved in time $2^d \cdot \poly(n, \log(t^+ + 2))$ (\zcref{prop:base}).

\subsection{Our contributions}
Our main contribution is an algorithm for \garrival that runs asymptotically faster than all previously known algorithms.
We assume that the number of tokens at each vertex of the switch digraph is encoded in binary, and that the total number of tokens is at most~$2^n$.

\begin{theorem}
\label{thm:main}
There is a deterministic algorithm that solves \garrival on $n$-vertex switch digraphs with at most $2^n$ tokens in time $2^{\mathcal{O}(n^{1/3} \cdot \log^{2/3} n)}$.
In particular, \arrival can be solved in time $2^{\mathcal{O}(n^{1/3} \cdot \log^{2/3} n)}$.
\end{theorem}

Compared with the previous time bound 
$n^{\mathcal{O}(\sqrt{n})} = 2^{\mathcal{O}(\sqrt{n}\log n)}$ of \textcite{GartnerHH2021} and
\textcite{Haslebacher2025}, our bound reduces the exponent
from $\sqrt{n}\log n$ to $n^{1/3}\log^{2/3}n$, a factor of $n^{1/6}\log^{1/3}n$.

Our second contribution, which we use to prove \zcref{thm:main}, is an algorithm for the least fixed points of maps with unit increments.
A map has \emph{unit increments} if increasing one entry of its argument by one increases at most one entry of its value, by one, and leaves the others unchanged (\zcref{def:unit-increments}).
\begin{theorem}
\label{thm:lfp}
There is a deterministic algorithm that, given integers $k \ge 1$ and $N \ge 0$ and oracle access to a map $F \colon \{0, \ldots, N\}^k \to \{0, \ldots, N\}^k$ with unit increments, outputs the least fixed point of $F$ after $\mathcal{O}(k^3 \cdot \log(kN + 2))$ evaluations of $F$ and additional time $(2k \cdot \log(N + 2))^{\mathcal{O}(k)}$.
\end{theorem}

Maps with unit increments are monotone, and the known algorithms for fixed points of monotone maps~\cite{FearnleyPS2022,ChenL2022,ChenLY2026} find some fixed point of them, but their numbers of evaluations are powers of $\log N$ whose exponents grow linearly in $k$.

\subparagraph*{Paper outline.}
\zcref{sec:approach} explains our approach and techniques.
\zcref{sec:preliminaries} collects notation for digraphs, vectors, norms, $\ell_1$-half-spaces and fixed points.
\zcref{sec:multiarrival} defines \garrival and switching flows, and states two known bounds.
\zcref{sec:guessing} defines guessing maps, unit increments and the cap.
\zcref{sec:recursion} gives the block recursion and bounds its run time.
\zcref{sec:counting} searches for the least fixed point by cutting $\ell_1$-half-spaces, and \zcref{sec:deep} computes the guesses of this search deterministically, which proves \zcref{thm:lfp}.
\zcref{sec:proof} proves \zcref{thm:main}.

\section{Overview}
\label{sec:approach}
As is the case with the previous results, our algorithm relies on \emph{switching flows} (\zcref{def:switching-flow}) which, as shown by \textcite{DohrauGKMW2017}, are certificates for \arrival (\zcref{prop:switching_flows}).
The algorithm by \textcite[Section~3.2]{Haslebacher2025}, which proposes an alternative view on the original result by \textcite{GartnerHH2021}, starts by finding, in the given switch digraph $(G, \rho)$, a set $S$ of $s$ non-terminals such that every vertex of $G$ is at distance at most $d$ from $S \cup T$, where~$T$ denotes the set of terminals.
In other words, $S \cup T$ is a \emph{distance-$d$ dominating set}, for some choice of $s$ and $d$, and the algorithm ensures that $sd$ is at most (roughly) $n \cdot \log n$~\cite[Lemma~12]{GartnerHH2021}.
A better bound for $sd$ follows from a theorem of \fulltextcite[Theorem~1.5]{NguyenSS2024} on small $k$-kernels of digraphs with a spanning arborescence, which gives $sd$ at most about $n$ while also guaranteeing that $S$ is independent.
In \zcref{lem:pivots} we show a simplified version of their result achieving a similar bound.

The recursion of this algorithm~\cite[Algorithm~1]{Haslebacher2025} then defines $S$ to be the set of \emph{pivots}, that is, vertices which are recursively turned into terminals to generate new instances until a base case is reached.
In this recursion, as in ours, the base case is an instance where every vertex is at distance at most $d$ from a terminal, which is then solved in time $2^d \cdot \poly(n, \log(t^+ + 2))$ by \zcref{prop:base}.
In short, the recursion produces guesses on how many tokens will pass through each pivot, then routes the tokens in the new instance, recursively down to the base case, and verifies if the guess is correct.
This is the case if, at the end of the routing, each pivot received exactly as many tokens as it sent.
If this is not the case, the guess is adjusted, and a new routing is tested.

The recursion adjusts the guess of each pivot by a binary search~\cite[Section~3.1]{Haslebacher2025}, which takes polynomially many routings per pivot, since, for \arrival, the number of tokens stays below $2^{\mathcal{O}(n^2)}$ along the recursion.
In addition, the depth $s$ of the recursion needed to ensure that every vertex of $S$ is turned into a terminal implies that the total run time is $n^{\mathcal{O}(s)} \cdot 2^d$, up to polynomial factors.
Balancing the choices of $s$ and $d$, a total run time of $2^{\mathcal{O}(\sqrt n \cdot \log n)}$ is thus obtained.

The algorithm by \textcite{GartnerHH2021} works similarly, and also starts by finding a set $S$ of $s$ pivots such that $S \cup T$ is distance-$d$ dominating.
In this case, however, all pivots are guessed at once.
They show that, in any instance of \arrival on a switch digraph with~$n$ vertices, each vertex is visited at most $N = 2^n$ times by the token~\cite[Lemma~1]{GartnerHH2021}, and hence a guess can be understood as a vector $\mathbf{w} = (w_b)_{b \in S}$ with entries in $[0, N]$.
For each such~$\mathbf{w}$, the map $F$ then defines the numbers $F(\mathbf{w})$ of tokens arriving at the vertices of $S$ in the instance $\mathcal{I}(S \leftarrow \mathbf{w})$ of \garrival where each $b \in S$ starts with $w_b$ tokens.
Hence, in order to compute $F(\mathbf{w})$ from $\mathbf{w}$, a routing of the tokens is needed, again through the base case, and the goal is to find a guess $\mathbf{w}^*$ such that $F(\mathbf{w}^*) = \mathbf{w}^*$, that is, a \emph{fixed point} of $F$.
Since $F$ is monotone~\cite[Theorem~8]{GartnerHH2021}, Tarski fixed point algorithms~\cite{DangQY2020,FearnleyPS2022} can be applied in order to speed up the search, after capping the values of $F$ at $N$ so that it maps the box $[0, N]^s$ to itself; the fixed point that they find is then also a fixed point of $F$~\cite[Lemma~11]{GartnerHH2021}.
This strategy yielded the first subexponential algorithm for \arrival.

In our case, in addition to \zcref{lem:pivots} for the improved bound, we employ a similar strategy with one difference: instead of restricting the guesses to one pivot at a time, as done by \textcite{Haslebacher2025}, or guessing all $s$ pivots at once, as done by \mbox{\textcite{GartnerHH2021}}, we guess the pivots in blocks of size $k$.
This allows us to reduce the depth of the recursion by a factor of $k$ when compared to the first approach, and the search dimension from $s$ to $k$ when compared to the second.
On the other hand, this strategy requires a far more difficult search procedure that needs to operate fast in the higher dimension~$k$.
Compared to the first approach, the difficulty is the higher dimension, and compared to the second, it is that the searches are now nested over several levels of the recursion, so that their numbers of evaluations multiply.
Each search can then afford only a number of evaluations that is polynomial in~$k$ and $\log N$, whereas the number of evaluations of the known Tarski fixed point algorithms~\mbox{\cite{FearnleyPS2022,ChenL2022,ChenLY2026}} is a power of $\log N$ whose exponent grows linearly in $k$.
As in the analyses of both previous algorithms, proving \zcref{thm:main} is a balancing act between the parameters~$d,s$ and now $k$, amongst others.

\subparagraph*{Block solvers.}
The map $F$ of \textcite{GartnerHH2021} is what we call the \emph{guessing map}~$\gmap_S$ of the pivots $S$ (\zcref{def:guessing-map}), and its fixed points are exactly the correct guesses (\zcref{lem:consistent}).

In our recursion, we split $S$ into $M = \lceil s/k \rceil$ \emph{blocks} $B_1, \ldots, B_M$ of at most $k$ pivots each, and we guess one block at a time (\zcref{algo:blocks}).
For a block $B$, the guessing map $\gmap_B$ is defined in the same way, with only the pivots of $B$ turned into terminals.
As with $F$~\cite{GartnerHH2021}, we cap it at a number $N$ so that it maps the box $[0, N]^B$ to itself, and our search finds the least fixed point of the capped map (\zcref{lem:capped}).

The search for a block is done by a \emph{block solver} (\zcref{def:block-solver}), and every \emph{evaluation} of the guessing map of the block is a recursive call that solves the instance in which the pivots of the block have become terminals, with the remaining blocks still to be guessed.
Hence, the depth of the recursion is $M$ instead of $s$.
If a block solver makes at most $Q$ evaluations and spends time at most $I$ besides them, which we call its \emph{overhead}, then the total run time is, up to polynomial factors, at most
\begin{equation*}
	(Q + 1)^M \cdot (2^d + I)
\end{equation*}
(\zcref{lem:recursion}).
The first factor bounds the number of calls of the recursion, since each of the~$M$ levels multiplies it by at most $Q + 1$: the evaluations of the block solver, and one last call at the fixed point that it finds.
The second factor bounds the time of a single call besides its recursive calls, which is the base case at the last level and the overhead of the block solver at the other levels.
With one pivot per block, as in the recursion of \textcite{Haslebacher2025}, a binary search for the least fixed point gives a polynomial $Q$ and a small overhead, but the depth is~$s$.
For larger blocks, the known Tarski fixed point algorithms~\cite{DangQY2020,FearnleyPS2022}, used in this way, give no asymptotic gain~\cite[Section~3]{Haslebacher2025}.
Our block solver, on the other hand, makes $Q = \mathcal{O}(k^3 \cdot \log(kN + 2))$ evaluations with overhead $I = (2k \cdot \log(N + 2))^{\mathcal{O}(k)}$ (\zcref{thm:lfp}), which is fast enough to pay off.
The next steps are to bound the impact of each of the parameters involved: the distance $d$ through the base case and, since $s \le n/(d + 1)$ (\zcref{lem:pivots}), through the number $M$ of blocks, the block size $k$ through $M$ and through the costs $Q$ and $I$ of the block solver, and the cap $N$, which grows along the recursion, through~$Q$ and $I$.
See~\zcref{fig:parameters} for a guide on the dependence of the parameters.

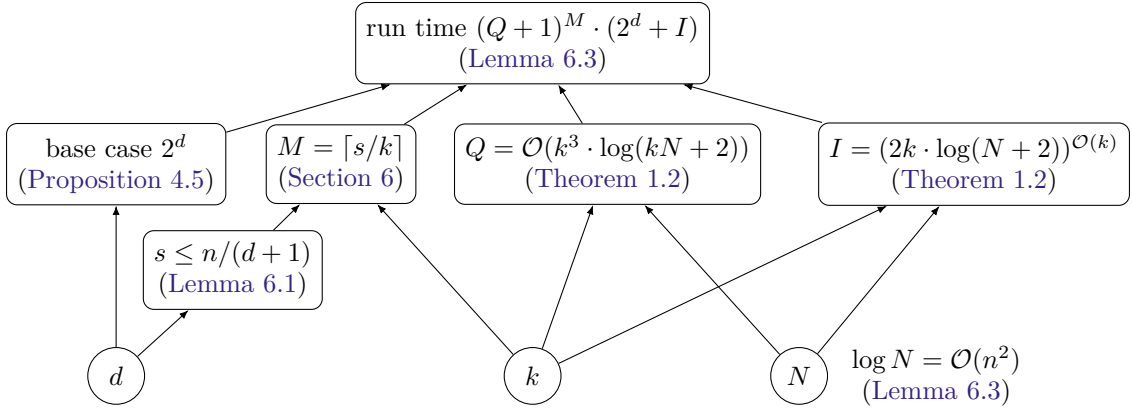
\begin{figure}[ht]
	\centering
	\begin{tikzpicture}[
        xscale=1.1,
		quantity/.style={draw, rounded corners, inner sep=4pt, font=\footnotesize, align=center},
		parameter/.style={draw, circle, inner sep=2pt, minimum size=0.75cm, font=\footnotesize}]
		\node[quantity] (time) at (0, 3) {run time $(Q + 1)^M \cdot (2^d + I)$\\(\zcref{lem:recursion})};
		\node[quantity] (base) at (-5, 1.4) {base case $2^d$\\(\zcref{prop:base})};
		\node[quantity] (M) at (-2.3, 1.4) {$M = \lceil s/k\rceil$\\(\zcref{sec:recursion})};
		\node[quantity] (Q) at (0.9, 1.4) {$Q = \mathcal{O}(k^3 \cdot \log(kN + 2))$\\(\zcref{thm:lfp})};
		\node[quantity] (I) at (5.3, 1.4) {$I = (2k \cdot \log(N + 2))^{\mathcal{O}(k)}$\\(\zcref{thm:lfp})};
		\node[quantity] (s) at (-3.6, 0) {$s \le n/(d + 1)$\\(\zcref{lem:pivots})};
		\node[parameter] (d) at (-5, -1.4) {$d$};
		\node[parameter] (k) at (0, -1.4) {$k$};
		\node[parameter] (N) at (3.2, -1.4) {$N$};
		\node[font=\footnotesize, right, align=center] at (3.7, -1.4) {$\log N = \mathcal{O}(n^2)$\\(\zcref{lem:recursion})};
		\foreach \a/\b in {base/time, M/time, Q/time, I/time, d/base, d/s, s/M, k/M, k/Q, k/I, N/Q, N/I} {
			\draw[arrow] (\a) -- (\b);
		}
	\end{tikzpicture}
	\caption{How the parameters enter the run time of the recursion. An arrow from one quantity to another means that the second one depends on the first. The distance $d$ and the block size $k$ are chosen in \zcref{sec:proof}, and the cap $N$, which grows along the recursion, satisfies $\log N = \mathcal{O}(n^2)$ when $t^+ \le 2^n$.}
	\label{fig:parameters}
\end{figure}

\subparagraph*{The base case and the pivots.}
We start with the distance $d$, which enters the run time through the base case and, via the size $s$ of the set of pivots, through the number $M$ of blocks.
Since the base case takes time exponential in $d$, the distance~$d$ cannot be large, but keeping $d$ small requires many pivots.
Our set $S$ of pivots is simple to obtain (\zcref{lem:pivots}): the non-terminals fall into $d + 1$ classes according to their distance to the set $T$ of terminals modulo $d + 1$, and $S$ is a smallest class, so that $s \le n/(d + 1)$.
Every vertex at distance $h > d$ from $T$ is then within distance $d$ of $S$, since the first $d + 1$ vertices of a shortest path from it to $T$ are non-terminals at the consecutive distances $h, \ldots, h - d$ from $T$, and one of them lies in $S$.
Hence, the number of blocks is $M = \lceil s/k \rceil \le n/(k(d + 1)) + 1$, and a larger $d$ reduces the depth of the recursion at the price of a more expensive base case.

\subparagraph*{The cap.}
We now choose the cap $N$ for the block solver: for a block $B$, it searches guesses in $[0,N]^B$, so this box must contain a consistent guess.
Let $\tau$ be the number of tokens of the instance in which $B$ is guessed.
By a result of \textcite[Lemma~2]{Haslebacher2025}, every switching flow of this instance sends fewer than $2^{n + 1} \cdot \tau$ tokens out of every vertex (\zcref{prop:flow-bound}), and hence, with $N = 2^{n + 1} \cdot \tau$, the correct guess read off from a switching flow lies in $[0, N]^B$, where none of its entries reaches the cap.
This is why the least fixed point of the capped map, which lies below this guess, is correct (\zcref{lem:capped}).
The price of the cap is that $N$ grows with the number of tokens, which grows along the recursion, since every evaluation turns the pivots of a block into terminals with up to $N$ starting tokens each.
Each level multiplies the number of tokens by at most $1 + k \cdot 2^{n + 1} \le 2^{2n + 1}$, so, over the at most $n$ levels, $\log N = \mathcal{O}(n^2)$ in every instance of the recursion when $t^+ \le 2^n$ (\zcref{lem:recursion}).
This is where \zcref{thm:main} uses its bound on the number of tokens.
On the other hand, the search depends on $N$ only mildly: as we point out below, the number $Q = \mathcal{O}(k^3 \cdot \log(kN + 2))$ of evaluations satisfies $\log Q = \mathcal{O}(\log n)$, and the overhead $I = (2k \cdot \log(N + 2))^{\mathcal{O}(k)}$ satisfies $\log I = \mathcal{O}(k \cdot \log n)$, which is the term that limits the block size $k$.

\subparagraph*{Searching by cutting half-spaces.}
For a block $B$ of $k$ pivots, the block solver searches the box $[0, N]^B$ for the least fixed point of the capped guessing map $H \colon [0, N]^B \to [0, N]^B$, which is the correct guess for the pivots of $B$ (\zcref{lem:capped}).
The only property of $H$ that the search uses is that it has \emph{unit increments} (\zcref{def:unit-increments}).
In other words, when a pivot sends one more token, this token either arrives at a single pivot, whose entry of $H$ grows by one, or it is absorbed, at a terminal of $T$ or at a pivot whose entry is capped.
This follows from the monotonicity of the arrivals in the starting tokens~\cite[Lemma~4]{Haslebacher2025} (\zcref{prop:monotone,lem:unit,lem:capped}).
In particular, $H(\mathbf{x})$ and $H(\mathbf{y})$ differ by at most as many tokens as $\mathbf{x}$ and $\mathbf{y}$.
That is, $\|H(\mathbf{x}) - H(\mathbf{y})\|_1 \le \|\mathbf{x} - \mathbf{y}\|_1$ (\zcref{def:norm} and~\eqref{eq:non-expansive}).
The search generalises binary search to $k$ dimensions.
It keeps a set $C$ of candidate guesses, initially the whole box, and in every round it evaluates $H$ at a guess $\mathbf{w}$.
If the \emph{surplus} $\surplus(\mathbf{w}) \coloneqq H(\mathbf{w}) - \mathbf{w}$ is zero, then $\mathbf{w}$ is a fixed point.
Otherwise, every pivot should send more tokens if it receives more than it sends, and fewer if it receives fewer.
Instead of moving~$\mathbf{w}$ towards~$H(\mathbf{w})$, we remove from $C$ the guesses on the \emph{wrong side} of $\mathbf{w}$, namely those whose $\ell_1$-distance to $\mathbf{w}$ does not shrink when $\mathbf{w}$ moves slightly towards $H(\mathbf{w})$.
These guesses form the \emph{$\ell_1$-half-space} $\halfspace{\mathbf{w}}{-\surplus(\mathbf{w})}$ in the sense of \fulltextcite[Definition~3.2]{HaslebacherLSW2025} (\zcref{def:half-space}), and the goal is to ensure that every such removal discards a meaningful fraction of the candidates.
These authors search in the same way for approximate fixed points of $\ell_1$-contractions~\cite[Section~4]{HaslebacherLSW2025}, that is, of maps that shrink $\ell_1$-distances by a constant factor.
The connection of \garrival to this search is due to \textcite[Section~4]{Haslebacher2025}, who reduces \garrival to finding such fixed points and observes that this search is efficient in the number of evaluations, but not in time~\cite[Section~1]{Haslebacher2025}.
For a single pivot, this is exactly the usual binary search, since $H(w) > w$ rules out every guess at or below $w$, and $H(w) < w$ every guess at or above~$w$.
The removal is almost safe, since a fixed point~$\mathbf{q}$ is always at least as close to $H(\mathbf{w})$ as to~$\mathbf{w}$, because $\mathbf{q} = H(\mathbf{q})$ and $H$ does not increase $\ell_1$-distances~\eqref{eq:non-expansive}, whereas the removed guesses are at least as close to $\mathbf{w}$ as to $H(\mathbf{w})$.
So $\mathbf{q}$ is removed only if both distances are equal.

The problem is that equality does occur, at guesses $\mathbf{w}$ whose surpluses sum to zero but that are \emph{not} fixed points, which we call \emph{ties}.
At a tie, the pivots receive in total exactly as many tokens as they send, but not every pivot receives as many tokens as it sends, because the tokens by which $\mathbf{w}$ and $\mathbf{q}$ differ only circulate among the pivots and never reach a terminal or the cap (\zcref{lem:tie} and the remark after it).
A round at a tie may remove every fixed point, as the map that swaps the tokens of two pivots shows (\zcref{sec:half-spaces}).
We stop the circulation with a leak: the search runs on a \emph{leaky map}, in which a small fraction of the tokens of every pivot disappears (\zcref{sec:sinks}).
Circulating tokens are then absorbed after a while, so a fixed point of the leaky map can be removed only in a round whose guess already rounds to the least fixed point~$\mathbf{q}^*$ of $H$, which the round records, and the search returns $\mathbf{q}^*$ (\zcref{lem:sinks,lem:cutting}).
The price paid is small: the search runs on the finer box $[0, \lambda N]^k$ with $\lambda = 16(kN)^2$, whose side still has a logarithm of $\mathcal{O}(\log(kN + 2))$, and every evaluation of the leaky map costs at most $k + 1$ evaluations of $H$ (\zcref{lem:sinks}).

The choice of the guesses has a direct impact on the run time of the algorithm.
We say that a guess $\mathbf{w}$ is \emph{$\alpha$-deep} in $C$ if every $\ell_1$-half-space through $\mathbf{w}$ contains at least an $\alpha$ fraction of $C$.
Then, independently of the surplus at $\mathbf{w}$, the round removes at least $\alpha|C|$ candidates when a fixed point is not found, and the \emph{$\ell_1$-centerpoint theorem} of \textcite[Theorem~3.18]{HaslebacherLSW2025} states that every finite set of points has a $\frac1{k + 1}$-deep point (\zcref{thm:hlsw}).
For a set of integer guesses, such as the candidates, rounding this point down keeps this property (\zcref{lem:rounding}).
Half of this depth suffices, and with $\frac1{2(k + 1)}$-deep guesses we obtain the count of evaluations promised in the part on the cap.
Since $C$ starts with at most $(\lambda N + 1)^k$ candidates, and every round removes at least a $\frac1{2(k + 1)}$ fraction of them, the search stops after $\mathcal{O}(k(k + 1) \cdot \log(\lambda N + 1)) = \mathcal{O}(k^2 \cdot \log(kN + 2))$ rounds.
With at most $k + 1$ evaluations of~$H$ per round, it makes $Q = \mathcal{O}(k^3 \cdot \log(kN + 2))$ evaluations (\zcref{lem:cutting,thm:lfp}).
It remains to find such guesses, which we discuss next.

\subparagraph*{Deep guesses.}
We now come to the heart of our algorithm and the main technical contribution of this paper: computing, in every round of the search, a $\frac1{2(k + 1)}$-deep guess, deterministically (\zcref{sec:deep}).
In one dimension, a $\frac12$-deep guess is a median of the candidates, and deep guesses play the role of medians in $k$ dimensions.
To compute them, we follow a simple scheme that is known for the usual half-spaces~\cite{Matousek1991center}, \mbox{\cite[Lemma~47.3.1]{MustafaV2017}}.
To find a point in the middle of a large set, that is, a point through which every half-space contains a large fraction of the set, one takes a small sample of the set that represents it, in the sense that every half-space contains about the same fraction of the sample as of the set.
Then a point in the middle of the sample is found by brute force, which is affordable since the sample is small, and this point is also in the middle of the set, since every half-space sees the sample and the set in about the same proportions.
A large enough uniformly random sample, whose size does not depend on the size of the set, represents the set with high probability~\cite[Theorem~47.2.2]{MustafaV2017}, and \textcite{ChazelleM1996} compute such a sample deterministically.

In our setting, the set is the set $C$ of candidates, the half-spaces are the $\ell_1$-half-spaces, and a point in the middle is a deep point, which we round down to a guess.
Our version of these steps rests on one property of $\ell_1$-half-spaces: whether a point lies in an $\ell_1$-half-space through $\mathbf{c}$ depends only on which of its coordinates are below, equal to or above those of $\mathbf{c}$ (\zcref{lem:orthant}).
\zcref{sec:centerpoints} shows that deep points exist, by the $\ell_1$-centerpoint theorem, and that rounding them down keeps them deep (\zcref{thm:hlsw,lem:rounding}).
\zcref{sec:exact} is the brute force: the $\ell_1$-half-spaces cut only $(12m^2)^k$ different subsets, their \emph{traces}, out of any $m$ points, and listing the traces gives an exact deepest point of few weighted points (\zcref{lem:traces,lem:exact-centerpoint}).
\zcref{sec:compression} gives the two steps that compute the sample.
After $r$ rounds, the candidates form a union of at most $(2r + 1)^k$ boxes (\zcref{lem:cutting}), but these few boxes contain far too many points for the deterministic sampling algorithm~\cite{ChazelleM1996}, whose time is linear in the number of points.
So we first cut every box into few sub-boxes and represent every sub-box by one point, weighted by its size, which changes the fraction of every $\ell_1$-half-space only slightly (\zcref{lem:compress}).
Then this algorithm~\cite[Theorem~2.3]{ChazelleM1996}, which needs the traces to be few and to be listable, thins these weighted points out to $\mathcal{O}(k^3 \cdot \log(k + 1))$ points that still represent the candidates (\zcref{prop:cm}).
\zcref{sec:solver} combines these steps (\zcref{lem:deep-point}).

The result is a $\frac1{2(k + 1)}$-deep guess for the candidates, computed in time \mbox{$(2k \cdot \log(N + 2))^{\mathcal{O}(k)}$} per round.
Computing it is where the overhead $I$ of the block solver comes from, and since the number of rounds is polynomial in $k$ and~$\log N$, the overhead of the search is of the same order, $I = (2k \cdot \log(N + 2))^{\mathcal{O}(k)}$ (\zcref{thm:lfp}).
In short, the block solver pays for its few evaluations with a large overhead, but only once per call, and the next part balances this overhead against the depth of the recursion.

\subparagraph*{The balancing act.}
We can now put the pieces together, following \zcref{fig:parameters}.
By \zcref{lem:recursion}, the run time of the recursion with block solvers is at most $(Q + 1)^M \cdot (2^d + I)$ up to polynomial factors, where $M$ is the number of blocks, that is, the depth of the recursion,~$Q$ is the number of evaluations that the block solver makes for one block, $I$ is its overhead, and~$2^d$ is the time of the base case, for the distance $d$ of every vertex to the nearest terminal or pivot.
Hence, the logarithm of the run time is at most $M \cdot \log(Q + 1) + d + \log I$ up to an additive $\mathcal{O}(\log n)$ term.
The previous parts bound each of these terms: the number of blocks satisfies $M \le n/(k(d + 1)) + 1$, while $\log(Q + 1) = \mathcal{O}(\log n)$ and $\log I = \mathcal{O}(k \cdot \log n)$, since $\log N = \mathcal{O}(n^2)$ when $t^+ \le 2^n$.
Hence, up to a constant factor and an additive $\mathcal{O}(\log n)$ term, the logarithm of the run time is
\begin{equation*}
	\frac{n \cdot \log n}{kd} + d + k \cdot \log n\enspace.
\end{equation*}
The first term, which counts the levels of the recursion, asks for large blocks and a large distance $d$, the second, which comes from the base case, asks for a small $d$, and the third, which comes from the overhead of the block solver, asks for small blocks.
The three terms balance for $d = n^{1/3} \cdot \log^{2/3} n$ and $k = d/\log n$, where each of them is of the order $n^{1/3} \cdot \log^{2/3} n$, and this proves \zcref{thm:main} (\zcref{sec:proof}).
\section{Preliminaries}
\label{sec:preliminaries}

\subparagraph*{Digraphs.}
For basic background on graph theory we refer the reader to \textcite{BondyM2008}.
For a digraph $G$, let $V(G)$ and $E(G)$ denote its sets of vertices and arcs, respectively.
For an arc $e = (u, v) \in E(G)$, we call $u$ its \emph{tail} and $v$ its \emph{head}, and we say that $e$ is an \emph{out-arc} of $u$.
We treat $E(G)$ as a multiset, so that two arcs may have the same tail and the same head.
Unless stated otherwise, digraphs have no loops, and $n$ denotes the number of vertices of the input digraph.
All paths in this paper are directed.
For $u, v \in V(G)$, the \emph{distance from $u$ to~$v$}, denoted $\dist(u, v)$, is the minimum length of a path from~$u$ to $v$, and $\infty$ if there is none.
For $X \subseteq V(G)$, the \emph{distance from $u$ to $X$} is $\dist(u, X) \coloneqq \min_{v \in X}\dist(u, v)$.
A vertex~$u$ is \emph{within distance~$d$ of} $X$ if $\dist(u, X) \le d$, and $X$ is \emph{distance-$d$ dominating} if every vertex of~$G$ is within distance $d$ of $X$.

\subparagraph*{Vectors.}
We denote by $\N$ and $\N_0$ the sets of positive and of non-negative integers, respectively.
For integers $a$ and $b$, we denote by $[a, b]$ the set $\{a, a + 1, \ldots, b\}$ of the integers from~$a$ to $b$, which is empty if $b < a$.
For $s \in \N_0$, we write $[s] \coloneqq [1, s]$.
We write vectors in bold and their entries in plain type, as in $\mathbf{w} = (w_b)_{b \in B}$.
The entries are indexed by a finite set, such as $[k]$ or a set $B$ of vertices, which will be clear from the context.
Inequalities between vectors hold entrywise.
We denote by $\0$ and $\1$ the vectors with only zeros and only ones, respectively.
For an index $i$, the $i^{\textnormal{th}}$ \emph{unit vector}~$\e_i$ has the $i^{\textnormal{th}}$ coordinate one and all other entries zero.
For a vector $\mathbf{z}$ of real numbers,~$\lfloor\mathbf{z}\rfloor$ is obtained by rounding every entry down to an integer, and $\sgn(\mathbf{z})$ is the vector of the signs $-1$, $0$ or $1$ of its entries.
A \emph{box} is a product $\prod_i[a_i, b_i]$ of sets of consecutive integers, one for every index $i$, which we call its \emph{sides}.
The \emph{length} of the side $[a_i, b_i]$ is $b_i - a_i + 1$.
For a finite non-empty set $X \subset \mathbb{R}^k$, the \emph{bounding box} $\bbox(X)$ is the smallest set of the form $\{\mathbf{z} \in \mathbb{R}^k : \mathbf{a} \le \mathbf{z} \le \mathbf{b}\}$, with $\mathbf{a}, \mathbf{b} \in \mathbb{R}^k$, that contains~$X$.

\subparagraph*{Norms and distances.}
For the remainder of this section, we assume that vectors have $k$ real entries.
We denote by $\mathbf{x} \cdot \mathbf{y}$ the \emph{dot product} $\sum_{i \in [k]}x_i y_i$  of $\mathbf{x}$ and $\mathbf{y}$.

\begin{definition}[$\ell_1$-norm and $\ell_1$-distance]
\label{def:norm}
The \emph{$\ell_1$-norm}, or \emph{Manhattan norm}, of a vector $\mathbf{x}$ is $\|\mathbf{x}\|_1 \coloneqq \sum_{i \in [k]}|x_i|$, and the \emph{$\ell_1$-distance}, or \emph{Manhattan distance}, between two vectors $\mathbf{x}$ and $\mathbf{y}$ is $\|\mathbf{x} - \mathbf{y}\|_1$.
\end{definition}

In our context, the vectors are guesses of the numbers of tokens that some vertices send, and the $\ell_1$-distance between two guesses is the total number of tokens by which they differ, summed over these vertices.
The $\ell_1$-distance satisfies the triangle inequality $\|\mathbf{x} - \mathbf{z}\|_1 \le \|\mathbf{x} - \mathbf{y}\|_1 + \|\mathbf{y} - \mathbf{z}\|_1$, since $|x_i - z_i| \le |x_i - y_i| + |y_i - z_i|$ for every $i$.
We also use the \emph{Euclidean norm} $\|\mathbf{x}\|_2 \coloneqq \sqrt{\mathbf{x} \cdot \mathbf{x}}$, but only to compare $\ell_1$-half-spaces with the usual half-spaces and to quote the $\ell_1$-centerpoint theorem~\cite{HaslebacherLSW2025} in \zcref{sec:centerpoints}.

\subparagraph*{Half-spaces.}
The usual closed half-space through a point $\mathbf{c}$ in a direction $\mathbf{v} \ne \0$ is the set $\{\mathbf{z} : \mathbf{v} \cdot (\mathbf{z} - \mathbf{c}) \ge 0\}$ of the points ahead of $\mathbf{c}$ in direction $\mathbf{v}$.
A point~$\mathbf{z}$ lies in it if and only if $\mathbf{z}$ does not come closer to $\mathbf{c}$, in the Euclidean distance, when~$\mathbf{c}$ is moved slightly backwards, to $\mathbf{c} - \varepsilon\mathbf{v}$ for a small $\varepsilon > 0$.
This follows from $\|\mathbf{c} - \varepsilon\mathbf{v} - \mathbf{z}\|_2^2 = \|\mathbf{c} - \mathbf{z}\|_2^2 + 2\varepsilon\,\mathbf{v} \cdot (\mathbf{z} - \mathbf{c}) + \varepsilon^2\|\mathbf{v}\|_2^2$.
\textcite[Definition~3.2]{HaslebacherLSW2025} use this description with other distances in place of the Euclidean distance, and we use it with the $\ell_1$-distance.

\begin{definition}[$\ell_1$-half-spaces~{\cite[Definition~3.2]{HaslebacherLSW2025}}]
\label{def:half-space}
  For a point $\mathbf{c} \in \mathbb{R}^k$ and a direction $\mathbf{v} \in \mathbb{R}^k \setminus \{\0\}$, the \emph{$\ell_1$-half-space through $\mathbf{c}$ in direction $\mathbf{v}$} is
  \begin{equation*}
	\halfspace{\mathbf{c}}{\mathbf{v}} \coloneqq \{\mathbf{z} \in \mathbb{R}^k : \|\mathbf{c} - \mathbf{z}\|_1 \le \|\mathbf{c} - \varepsilon\mathbf{v} - \mathbf{z}\|_1\text{ for all }\varepsilon > 0\}\enspace.
  \end{equation*}
\end{definition}

The set $\halfspace{\mathbf{c}}{\mathbf{v}}$ does not change when $\mathbf{v}$ is multiplied by a positive number.
Unlike a usual half-space, it is in general not bounded by a single hyperplane (\zcref{lem:orthant} and \zcref{fig:half-space}).

\subparagraph*{Fixed points.}
A \emph{fixed point} of a map $F$ is an element $\mathbf{x}$ of its domain with $F(\mathbf{x}) = \mathbf{x}$.
A map $F$ defined on a set $X$ of vectors is \emph{monotone} if $F(\mathbf{x}) \le F(\mathbf{y})$ for all $\mathbf{x}, \mathbf{y} \in X$ with $\mathbf{x} \le \mathbf{y}$.
Our searches are for fixed points of monotone maps $F \colon [0, N]^k \to [0, N]^k$.
By Tarski's fixed point theorem~\cite[Theorem~1]{Tarski1955}, such a map has a least fixed point.
We use the following form of this fact, in which the least fixed point is found by iterating the map from the corner $\0$.

\begin{lemma}
\label{lem:kleene}
  Let $N \in \N_0$, and let $F \colon [0, N]^k \to [0, N]^k$ be monotone.
  Then the sequence $\0, F(\0), F(F(\0)), \ldots$ is non-decreasing, it becomes stationary after at most~$kN$ steps, and its final value $\mathbf{q}^*$ is a fixed point of $F$ with $\mathbf{q}^* \le \mathbf{q}$ for each fixed point~$\mathbf{q}$ of~$F$.
\end{lemma}
\begin{proof}
  Let $\mathbf{u}_0 \coloneqq \0$ and $\mathbf{u}_{i + 1} \coloneqq F(\mathbf{u}_i)$ for $i \ge 0$.
  Then $\mathbf{u}_0 \le \mathbf{u}_1$, since $\mathbf{u}_1 \in [0, N]^k$.
  If $\mathbf{u}_i \le \mathbf{u}_{i + 1}$, then $\mathbf{u}_{i + 1} = F(\mathbf{u}_i) \le F(\mathbf{u}_{i + 1}) = \mathbf{u}_{i + 2}$ by monotonicity.
  So the sequence is non-decreasing.
  Every step that changes the vector increases the sum of its entries by at least one, and this sum is at most $kN$ in $[0, N]^k$.
  So the sequence becomes stationary after at most $kN$ steps, at a vector $\mathbf{q}^*$ with $F(\mathbf{q}^*) = \mathbf{q}^*$.
  For every fixed point $\mathbf{q}$ of $F$, we have $\mathbf{u}_0 = \0 \le \mathbf{q}$, and $\mathbf{u}_i \le \mathbf{q}$ gives $\mathbf{u}_{i + 1} = F(\mathbf{u}_i) \le F(\mathbf{q}) = \mathbf{q}$.
  So $\mathbf{q}^* \le \mathbf{q}$.
\end{proof}

We call $\mathbf{q}^*$ the \emph{least fixed point} of $F$.

\section{The \garrival problem}
\label{sec:multiarrival}

\label{sec:instances}
A \emph{switch digraph} is a digraph $G$ in which every vertex has at most two out-arcs, together with a labelling $\rho$.
At every vertex with two out-arcs, $\rho$ marks one of them as the \emph{first} and the other one as the \emph{second} out-arc.
In figures, the first out-arc of a vertex is black (solid) and the second one is red (dashed).

A \garrival instance~\cite{Hoang2022,Haslebacher2025} is a tuple $\mathcal{I} = (G, \rho, T, (t^+_v)_{v \in T})$, where $(G, \rho)$ is a switch digraph, $T \subseteq V(G)$ is the non-empty set of \emph{terminals}, and $t^+_v \in \N_0$ is the number of tokens that start at the terminal $v$.
The total number of tokens is $t^+ \coloneqq \sum_{v \in T}t^+_v$.
We assume throughout that every non-terminal reaches a terminal, since otherwise tokens might circulate forever.
We also assume that every terminal $v$ with $t^+_v \ge 1$ has an out-arc, since otherwise it cannot send its tokens, and that $t^+ \ge 1$.
In particular, $n \ge 2$, since a terminal with tokens has an out-arc, which is not a loop.

The tokens move as follows, which we call the \emph{routing} procedure.
Initially, every terminal~$v$ with $t^+_v \ge 1$ sends its tokens.
If $v$ has two out-arcs, it sends $\lceil t^+_v/2\rceil$ tokens along its first and $\lfloor t^+_v/2\rfloor$ tokens along its second out-arc.
If $v$ has one out-arc, it sends all its tokens along it.
After this, terminals send no tokens, and every token that reaches a terminal stays there.
This convention follows \textcite{Haslebacher2025}.
While some non-terminal $v$ holds a token, one token at $v$ is moved along the out-arc that $v$ has used the fewest times so far, and along the first out-arc if both have been used equally often.
So a non-terminal with two out-arcs sends its tokens along them alternately, starting with the first one.
The order in which the tokens at different non-terminals are moved is free.

\begin{definition}[Runs and run profiles]
\label{def:run}
A \emph{run} of the routing procedure is a sequence of moves that follows these rules and that continues as long as some non-terminal holds a token, and it \emph{ends} once all tokens have reached terminals.
The \emph{run profile} of a run that ends records, for every arc, how many tokens traverse it.
\end{definition}

The task of \garrival is to determine, for each terminal $v \in T$, the number~$t^-_v$ of tokens that arrive at $v$.
Every run ends, whatever the order in which the tokens are moved, and all runs produce the same numbers $t^-_v$ (\zcref{prop:switching_flows} below).
The known proofs of \zcref{prop:switching_flows}~\cite{DohrauGKMW2017,GartnerHH2021} use switching flows, which we now define.
For a function $x \colon E(G) \to \N_0$ and a vertex $v$, we denote by $x^+(v)$ the sum of $x(e)$ over the out-arcs $e$ of $v$, and by $x^-(v)$ the sum of $x(e)$ over the arcs $e$ with head $v$.
\begin{definition}[Switching flows~{\cite{DohrauGKMW2017,GartnerHH2021}}]
\label{def:switching-flow}
A \emph{switching flow} of the instance $\mathcal{I}$ is a function $x \colon E(G) \to \N_0$ such that
\begin{align}
	x(e_1) = \lceil x^+(v)/2\rceil \text{ and } x(e_2) = \lfloor x^+(v)/2\rfloor & \qquad\text{for $v$ with out-arcs $e_1, e_2$} \label{eq:rotor}\\
	x^+(v) = x^-(v) & \qquad\text{for } v \in V(G) \setminus T, \label{eq:conservation}\\
	x^+(v) = t^+_v & \qquad\text{for } v \in T\enspace, \label{eq:terminals}
\end{align}
where $e_1$ and $e_2$ denote the first and the second out-arc of $v$.
\end{definition}

A switching flow need not be a run profile (\zcref{fig:fixed-point-example}).
\textcite{DohrauGKMW2017} proved \zcref{prop:switching_flows} for \arrival.
\textcite{GartnerHH2021} generalised it to \arrival instances in which some vertices become terminals with given numbers of tokens, and they also prove that every run ends~\cite[Lemma~6]{GartnerHH2021}.
The result is also stated for \garrival~\cite[Theorem~1]{Haslebacher2025}.

\begin{proposition}[\cite{DohrauGKMW2017}, {\cite[Lemma~6 and Theorem~8]{GartnerHH2021}}, {\cite[Theorem~1]{Haslebacher2025}}]
\label{prop:switching_flows}
  For every \garrival instance $\mathcal{I} = (G, \rho, T, (t^+_v)_{v \in T})$, every run of the routing procedure ends, and its run profile is a switching flow of $\mathcal{I}$.
  The number $t^-_v$ of tokens that arrive at a terminal $v \in T$ is the same in every run, and every switching flow $x$ of $\mathcal{I}$ satisfies $x^-(v) = t^-_v$ for every $v \in T$.
\end{proposition}

We use Theorem~1 and Lemmas~2, 4 and~6 of \textcite{Haslebacher2025} as \zcref{prop:switching_flows}, \zcref{prop:flow-bound}, \zcref{prop:monotone} and \zcref{prop:base}.
They are stated there for switch digraphs in which every vertex has exactly two out-arcs.
Our instances can be brought to this form without changing the routing.
Every vertex with one out-arc $(v, u)$ gets a second arc parallel to it, so that $v$ sends $\lceil m/2\rceil$ and $\lfloor m/2\rfloor$ of its $m$ tokens along the two copies.
Every vertex without out-arcs, which is a terminal without tokens, gets two loops, which carry no tokens.
So the cited results apply to our instances.

An \arrival instance (\zcref{sec:introduction}) is a switch digraph $(G, \rho)$, possibly with loops, with two terminals $t$ and $t'$, which we call \emph{destinations}, and an origin $o$.
The initially active out-arc of every vertex is the first one.
Every vertex other than $t$ and $t'$ has two out-arcs and reaches a destination.
A single token starts at $o$ and moves along the out-arcs as in the routing procedure, and the question is which destination it reaches.
If a vertex $v \notin \{t, t'\}$ has a loop and another out-arc $(v, u)$, then $v$ uses its two out-arcs alternately, so the token leaves $v$ along $(v, u)$ at every visit, possibly after one traversal of the loop.
We therefore replace the loop by a second arc from~$v$ to~$u$.
No vertex outside $\{t, t'\}$ has two loops, since such a vertex would reach no destination.
The out-arcs of $t$ and $t'$, which may be loops, are irrelevant, so we delete them.
If~$o$ is a destination, the answer is $o$.
Otherwise, we add a new terminal $s$ with one out-arc $(s, o)$ and $t^+_s \coloneqq 1$, and we make $t$ and $t'$ terminals with no tokens.
This is a \garrival instance on $n + 1$ vertices, with out-degree at most two, that satisfies our assumptions.
The destination at which the token arrives is the one with $t^-_v = 1$.

The next statement, due to Haslebacher, bounds the values of any switching flow.

\begin{proposition}[{\cite[Lemma~2]{Haslebacher2025}}]
\label{prop:flow-bound}
  Every switching flow $x$ of a \garrival instance on $n$ vertices satisfies $x^+(v) < 2^{n + 1} \cdot t^+$ for every vertex $v$.
\end{proposition}

The original statement~\cite[Lemma~2]{Haslebacher2025} says that every arc carries fewer than $2^n \cdot t^+$ tokens, and attributes this bound to \textcite{DohrauGKMW2017}.
After the change of our instances described above, every vertex has two out-arcs, which gives the bound.

The following proposition is the base case of our recursion.
It is stated for \garrival without proof~\cite[Lemma~6]{Haslebacher2025}, with an attribution to \textcite{GartnerHH2021}, who prove a version of it for \arrival by simulating the tokens~\cite[Lemma~10]{GartnerHH2021}.

\begin{proposition}[{\cite[Lemma~6]{Haslebacher2025}}]
\label{prop:base}
Let $\mathcal{I}$ be a \garrival instance on~$n$ vertices in which $T$ is distance-$d$ dominating.
Then the numbers $t^-_v$ for $v \in T$ can be computed in time $2^d \cdot \poly(n, \log(t^+ + 2))$.
\end{proposition}

\section{Guessing maps and the cap}
\label{sec:guessing}
Throughout this section, $\mathcal{I} = (G, \rho, T, (t^+_v)_{v \in T})$ is a \garrival instance on $n$ vertices.

\begin{definition}[Guessing maps and consistent guesses]
\label{def:guessing-map}
Let $B \subseteq V(G) \setminus T$ be a set of \emph{pivots}, and let $\mathbf{w} \in \N_0^B$ be a vector of \emph{guesses}.
Let $\mathcal{I}(B \leftarrow \mathbf{w})$ be the instance obtained from $\mathcal{I}$ by turning every pivot $b \in B$ into a terminal with $w_b$ starting tokens.
The \emph{guessing map} $\gmap_B$ of~$B$ sends $\mathbf{w}$ to the vector $\gmap_B(\mathbf{w}) \in \N_0^B$ of the numbers of tokens that arrive at the pivots in $\mathcal{I}(B \leftarrow \mathbf{w})$.
A guess $\mathbf{w}$ is \emph{consistent} if it is a fixed point of $\gmap_B$, that is, if every pivot receives as many tokens as it sends.
\end{definition}

The instance $\mathcal{I}(B \leftarrow \mathbf{w})$ satisfies the assumptions of \zcref{sec:instances}.
Its non-terminals reach~$T$, its tokens are those of $\mathcal{I}$ and the guessed ones, and every pivot has an out-arc, since it is a non-terminal of $\mathcal{I}$ that reaches $T$.
So $\gmap_B$ is well-defined, by \zcref{prop:switching_flows}.
\textcite[Theorem~8]{GartnerHH2021} use guessing maps for \arrival.
\zcref{fig:fixed-point-example} shows an example.

\begin{figure}[ht]
	\centering
	\begin{tikzpicture}
		\node[vertex, minimum size=.6cm, label=90:$s$] (s) at (-2,0) {$1$};
		\node[blackvertex, scale=1, label=240:$o$] (o) at (0,0) {};
		\node[blackvertex, scale=1, label=60:$a$] (a) at (2.5,0) {};
		\node[blackvertex, scale=1, label=270:$b$] (b) at (1.25,-1.5) {};
		\node[vertex, minimum size=.6cm, label=90:$t$] (D) at (4.5,0) {$0$};
		\node[vertex, minimum size=.6cm, label=0:$t'$] (Db) at (3.5,-1.5) {$0$};
		\draw[arrow] (s) to (o);
		\draw[arrow] (o) to [bend left = 20] (a);
		\draw[arrow] (a) to [bend left = 20] (o);
		\draw[arrow] (b) to (a);
		\draw[arrow, goodred,densely dashed] (o) to (b);
		\draw[arrow, goodred,densely dashed] (a) to (D);
		\draw[arrow, goodred,densely dashed] (b) to (Db);
	\end{tikzpicture}

	\medskip
	\begin{tabular}{r|cccccccc}
		$g$ & $0$ & $1$ & $2$ & $3$ & $4$ & $5$ & $6$ & $7$\\ \hline
		$\gmap_a(g)$ & $1$ & $2$ & $\mathbf{2}$ & $\mathbf{3}$ & $3$ & $3$ & $3$ & $4$\\
		tokens arriving at $t$ & $0$ & $0$ & $1$ & $1$ & $2$ & $2$ & $3$ & $3$\\
		tokens arriving at $t'$ & $0$ & $0$ & $0$ & $0$ & $0$ & $1$ & $1$ & $1$
	\end{tabular}
	\caption{An instance of \garrival encoding an instance of \arrival. The terminal~$s$ holds the token, and terminals $t,t'$ are the destinations, whose out-arcs are not drawn since they send no tokens. The token travels along $s, o, a, o, b, a, t$. The table shows the guessing map $\gmap_a \coloneqq \gmap_{\{a\}}$ of the single pivot $a$, with a guess written as a number $g \in \N_0$. If $a$ is turned into a terminal that sends $g$ tokens, then $\gmap_a(g)$ tokens arrive at $a$. The fixed points~$2$ and $3$ (bold) are the numbers of tokens that $a$ sends in the two switching flows of the instance. They are the run profile, and the run profile plus one token circulating once along $o, a, o$.}
	\label{fig:fixed-point-example}
\end{figure}
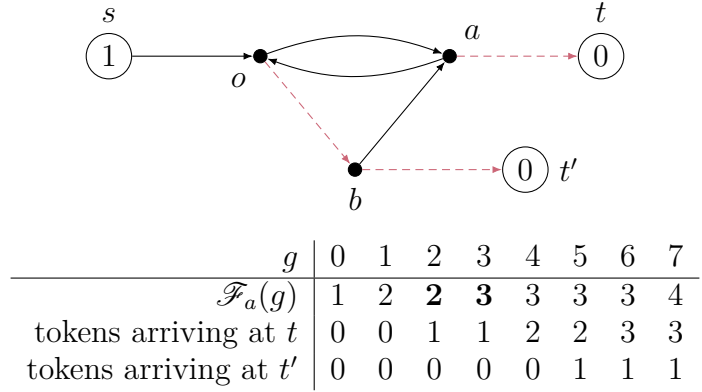

\begin{lemma}[cf.~{\cite[Lemma~11]{GartnerHH2021}} and~{\cite[Section~3.1]{Haslebacher2025}}]
\label{lem:consistent}
  Let $B \subseteq V(G) \setminus T$.
  If $\mathbf{w} \in \N_0^B$ is a consistent guess, then every terminal of $\mathcal{I}$ receives as many tokens in $\mathcal{I}(B \leftarrow \mathbf{w})$ as in $\mathcal{I}$.
  Conversely, for every switching flow $x$ of $\mathcal{I}$, the guess $(x^+(b))_{b \in B}$ is consistent.
\end{lemma}
\begin{proof}
  Let $\mathbf{w} \in \N_0^B$, and let $x \colon E(G) \to \N_0$ with $x^+(b) = w_b$ for every $b \in B$.
  Both~$\mathcal{I}$ and $\mathcal{I}(B \leftarrow \mathbf{w})$ ask for~\eqref{eq:rotor} at each vertex, \eqref{eq:conservation} at non-terminals of $\mathcal{I}$ outside~$B$, and~\eqref{eq:terminals} at the terminals of $\mathcal{I}$.
  Beyond this, $\mathcal{I}(B \leftarrow \mathbf{w})$ asks for $x^+(b) = w_b$ at the pivots, which holds, and $\mathcal{I}$ asks for~\eqref{eq:conservation} at the pivots.
  So if $x$ is a switching flow of $\mathcal{I}$ or of $\mathcal{I}(B \leftarrow \mathbf{w})$, then it is one of $\mathcal{I}(B \leftarrow \mathbf{w})$.
  Then $x^-(b) = \gmap_B(\mathbf{w})_b$ for every $b \in B$ by \zcref{prop:switching_flows}.
  Hence, $x$ is a switching flow of $\mathcal{I}$ if and only if it is one of $\mathcal{I}(B \leftarrow \mathbf{w})$ and $\gmap_B(\mathbf{w}) = \mathbf{w}$.

For the first claim, let $\mathbf{w}$ be consistent, and let $x$ be a switching flow of $\mathcal{I}(B \leftarrow \mathbf{w})$, such as the run profile of a run (\zcref{prop:switching_flows}).
Then $x$ is a switching flow of $\mathcal{I}$ as well.
By \zcref{prop:switching_flows}, applied to both instances, every terminal $v$ of $\mathcal{I}$ receives $x^-(v)$ tokens in both.
For the second claim, apply the equivalence to $\mathbf{w} \coloneqq (x^+(b))_{b \in B}$.
\end{proof}

\begin{definition}[Unit increments]
\label{def:unit-increments}
Let $B$ be a finite set, and let $X$ be a box in $\N_0^B$ or $X = \N_0^B$.
A map $F \colon X \to \mathbb{Z}^B$ has \emph{unit increments} if $F(\mathbf{w} + \e_b) - F(\mathbf{w}) \in \{\0\} \cup \{\e_{b'} : b' \in B\}$ whenever~$\mathbf{w}$ and $\mathbf{w} + \e_b$ are in $X$.
\end{definition}

In our context, $B \subseteq V(G) \setminus T$ is a set of pivots and $F$ is a guessing map.
Unit increments then say that when a pivot sends one more token, this token arrives at a terminal of $T$ or at a pivot $b' \in B$ whose entry grows by one.
A map with unit increments is monotone.
Indeed, let $\mathbf{x} \le \mathbf{y}$ be guesses of its domain.
A \emph{path of guesses} from $\mathbf{x}$ to $\mathbf{y}$ is a sequence of guesses from $\mathbf{x}$ to $\mathbf{y}$ that adds one token per step.
It has $\|\mathbf{y} - \mathbf{x}\|_1$ steps and stays in the domain.
By unit increments, no step decreases an entry of the map, so $F(\mathbf{x}) \le F(\mathbf{y})$.
Unit increments of guessing maps follow from the monotonicity of the arrivals in the starting tokens, which was first proved for \arrival~\cite[Theorem~8(ii)]{GartnerHH2021} and then for \garrival~\cite[Lemma~4]{Haslebacher2025}:

\begin{proposition}[{\cite[Lemma~4]{Haslebacher2025}}]
\label{prop:monotone}
Let $(G, \rho)$ be a switch digraph, and let $T \subseteq V(G)$ be a set of terminals.
For \garrival instances on $(G, \rho)$ with terminal set~$T$, the map that sends the vector $(t^+_v)_{v \in T}$ of starting tokens to the vector $(t^-_v)_{v \in T}$ of arriving tokens is monotone.
\end{proposition}

In words, if no terminal starts with fewer tokens, then no terminal receives fewer tokens.
For completeness, we derive unit increments from \zcref{prop:monotone}.

\begin{lemma}
\label{lem:unit}
Every guessing map $\gmap_B$ has unit increments.
\end{lemma}
\begin{proof}
Let $\mathbf{w} \in \N_0^B$ and $b \in B$.
The instances $\mathcal{I}(B \leftarrow \mathbf{w})$ and $\mathcal{I}(B \leftarrow \mathbf{w} + \e_b)$ have the same switch digraph and the same terminal set $T \cup B$, and the second one has one more starting token, at $b$.
By \zcref{prop:monotone}, no terminal receives fewer tokens in the second instance.
Every token arrives at a terminal, since every run ends (\zcref{prop:switching_flows}).
So the total number of arriving tokens is the number of starting tokens, which grows by exactly one.
Hence exactly one terminal $z \in T \cup B$ receives one more token, and every other terminal receives as many tokens as before.
So $\gmap_B(\mathbf{w} + \e_b) - \gmap_B(\mathbf{w})$ is $\e_z$ if $z$ is a pivot, and $\0$ if $z \in T$.
\end{proof}

Our searches need a map $H \colon [0, N]^B \to [0, N]^B$, and a guessing map need not be one.
For example, let $a$ and $b$ be pivots such that both out-arcs of $b$ lead to $a$, and let a terminal with one starting token have its only out-arc to $a$.
Then, for every guess~$\mathbf{w}$, the pivot $a$ receives at least $w_b + 1$ tokens, so $\gmap_B(N\1)_a > N$ for every $N$.
We therefore cap the guessing map at~$N$.
For a single pivot $b$, no cap is needed, since~$\gmap_{\{b\}}$ maps a large enough $[0, N]$ to itself~\cite[proof of Lemma~5]{Haslebacher2025}.
\textcite[Lemma~11]{GartnerHH2021} cap guessing maps of \arrival in the same way, at $2^n$.
They show that a solver for fixed points of monotone maps, run on the capped map, finds a consistent guess.
Our solver finds the least fixed point, and \zcref{lem:capped} shows briefly that it is consistent.

\begin{definition}[Capped guessing maps]
\label{def:capped}
  For a set $B$ of pivots and $N \in \N$, the \emph{capped guessing map} $H \colon [0, N]^B \to [0, N]^B$ is the guessing map bounded from above by $N\1$, that is, $H(\mathbf{w}) \coloneqq \min\{\gmap_B(\mathbf{w}), N\1\}$, with the minimum taken entrywise.
\end{definition}

\begin{lemma}
\label{lem:capped}
  Let $B \subseteq V(G) \setminus T$ be a set of pivots, and let $H$ be its capped guessing map for $N \coloneqq 2^{n + 1} \cdot t^+$.
Then $H$ has unit increments, and its least fixed point is a consistent guess.
\end{lemma}
\begin{proof}
  If $\mathbf{w}$ and $\mathbf{w} + \e_b$ lie in $[0, N]^B$, then $\gmap_B(\mathbf{w} + \e_b) - \gmap_B(\mathbf{w})$ is $\0$ or a unit vector~$\e_{b'}$ by \zcref{lem:unit}.
  So the capped entry $b'$ grows by at most one and no other entry changes, that is, $H(\mathbf{w} + \e_b) - H(\mathbf{w}) \in \{\0, \e_{b'}\}$.
Hence $H$ has unit increments, and so it is monotone.

  For the least fixed point, let $x$ be the run profile of a run on $\mathcal{I}$, which is a switching flow of $\mathcal{I}$ by \zcref{prop:switching_flows}.
  The guess $\mathbf{q} \coloneqq (x^+(b))_{b \in B}$ is consistent by \zcref{lem:consistent}, and $q_b < N$ for every $b \in B$ by \zcref{prop:flow-bound}.
  So $\mathbf{q}$ lies in $[0, N]^B$, and it is a fixed point of~$H$, as no entry of $\gmap_B(\mathbf{q}) = \mathbf{q}$ reaches the cap.
  By \zcref{lem:kleene}, the least fixed point $\mathbf{q}^*$ of $H$ satisfies $\mathbf{q}^* \le \mathbf{q} < N\1$.
  So no entry of $H(\mathbf{q}^*) = \mathbf{q}^*$ reaches the cap either, which means that $\gmap_B(\mathbf{q}^*) = H(\mathbf{q}^*) = \mathbf{q}^*$.
\end{proof}

\section{The block recursion}
\label{sec:recursion}
Throughout this section, $\mathcal{I} = (G, \rho, T, (t^+_v)_{v \in T})$ is a \garrival instance on $n$ vertices with $t^+ \le 2^n$ tokens, and $d, k \in [n]$.
We first choose the pivots.

\begin{lemma}
\label{lem:pivots}
  One can compute in polynomial time a set $S \subseteq V(G) \setminus T$ with $|S| \le n/(d + 1)$ such that $S \cup T$ is distance-$d$ dominating.
\end{lemma}
\begin{proof}
  Every non-terminal has a finite distance to $T$, since it reaches a terminal, and breadth-first search computes these distances.
  For $r \in [0, d]$, let $S_r$ be the set of non-terminals $v$ with $\dist(v, T) \equiv r \pmod{d + 1}$.
  These $d + 1$ sets partition the non-terminals, so a smallest one has at most $n/(d + 1)$ vertices, and we let $S \coloneqq S_r$ for an $r$ with $|S_r|$ smallest.
  Now let $v$ be a vertex with $h \coloneqq \dist(v, T)$.
  If $h \le d$, then~$v$ is within distance $d$ of $T$.
  Otherwise, a shortest path from $v$ to $T$ passes through vertices at distances $h, h - 1, \ldots, 0$ from $T$.
  Its first $d + 1$ vertices lie at the distances $h, \ldots, h - d \ge 1$ from $T$, so they are non-terminals, and $v$ is within distance $d$ of each of them.
  One of these $d + 1$ consecutive distances is congruent to~$r$ modulo $d + 1$, so one of these vertices lies in $S$, and $v$ is within distance $d$ of $S$.
\end{proof}

Let $S$ be the set computed in \zcref{lem:pivots}, which may be empty.
If $S=\varnothing$, then $M=0$ and \zcref{algo:blocks} immediately uses the base case.
We cut it, in any order, into $M \coloneqq \lceil|S|/k\rceil$ \emph{blocks} $B_1, \ldots, B_M$ of at most $k$ pivots each.
Then $M \le |S|/k + 1 \le n/(k(d + 1)) + 1$.
Every block is non-empty, since $(M - 1)k < |S|$.
For $j \in [0, M]$, let $T^{(j)} \coloneqq T \cup B_1 \cup \cdots \cup B_j$, so that $T^{(0)} = T$ and $T^{(M)} = T \cup S$.
The recursion runs from $T^{(0)}$ to $T^{(M)}$, and it guesses every block with an algorithm of the following kind.

\begin{definition}[Block solvers]
\label{def:block-solver}
A \emph{block solver} for block size $k$ is an algorithm that, given $k' \in [k]$, $N \in \N_0$ and oracle access to a map $H \colon [0, N]^{k'} \to [0, N]^{k'}$ with unit increments, returns the least fixed point of $H$.
Its costs are its number of evaluations of $H$ and its \emph{overhead}, the time that it spends besides the evaluations.
\end{definition}

A call of \zcref{algo:blocks} at \emph{level} $j$ receives an instance $\mathcal{I}_j$ on the switch digraph $(G, \rho)$ with terminal set $T^{(j)}$.
For $j = 0$, this is the input instance.
For $j > 0$, the vertices of $B_1, \ldots, B_j$ have been turned into terminals with guessed numbers of tokens.
To run the block solver for the block $B_{j + 1}$, a call at level $j < M$ identifies $[0, N]^{B_{j + 1}}$ with $[0, N]^{|B_{j + 1}|}$ by a fixed ordering of $B_{j + 1}$.
The recursion is that of \textcite[Algorithm~1 and Section~3.2]{Haslebacher2025}, with blocks instead of single pivots.
Our bookkeeping of the numbers of tokens in the proof of \zcref{lem:recursion} follows the analysis for single pivots~\cite[Lemma~5]{Haslebacher2025}.
The goal of \zcref{sec:counting,sec:deep} is to build a fast block solver.

\begin{algorithm}
	\caption{Find-Arrivals$(j, \mathcal{I}_j)$.}
	\label{algo:blocks}
	\begin{algorithmic}[1]
		\Require A level $j \in [0, M]$ and a \garrival instance $\mathcal{I}_j$ on $(G, \rho)$ with terminal set $T^{(j)}$. The parameters $d$ and $k$, the blocks $B_1, \ldots, B_M$ and a block solver $\mathcal{A}$ for block size $k$ are fixed.
		\Ensure The numbers $t^-_v$ of tokens that arrive at the terminals $v \in T^{(j)}$ of $\mathcal{I}_j$.
		\If{$j = M$} \Comment{Base case}
		\State \Return the numbers $t^-_v$ computed by the algorithm of \zcref{prop:base} on $\mathcal{I}_j$
		\EndIf
		\State $N \gets 2^{n + 1} \cdot \tau$, where $\tau$ is the number of tokens of $\mathcal{I}_j$ \label{line:alg2-cap}
		\State Run $\mathcal{A}$ on the map $H \colon [0, N]^{B_{j + 1}} \to [0, N]^{B_{j + 1}}$, where $H(\mathbf{w})$ is evaluated as follows:
		\State \qquad $(s_v)_{v \in T^{(j + 1)}} \gets$ Find-Arrivals$\bigl(j + 1, \mathcal{I}_j(B_{j + 1} \leftarrow \mathbf{w})\bigr)$ \label{line:alg2-call}
		\State \qquad $H(\mathbf{w}) \gets \min\{(s_b)_{b \in B_{j + 1}}, N\1\}$ \label{line:alg2-read}
		\State Let $\mathbf{w}^*$ be the least fixed point returned by $\mathcal{A}$
		\State $(s_v)_{v \in T^{(j + 1)}} \gets$ Find-Arrivals$\bigl(j + 1, \mathcal{I}_j(B_{j + 1} \leftarrow \mathbf{w}^*)\bigr)$ \label{line:alg2-final}
		\State \Return $(s_v)_{v \in T^{(j)}}$ \label{line:alg2-return}
	\end{algorithmic}
\end{algorithm}

\clearpage
\pagebreak
\begin{lemma}
\label{lem:recursion}
  Let $\mathcal{A}$ be a block solver for block size $k$ (\zcref{def:block-solver}).
  Suppose that, on every input with $N \le 2^{4n^2}$,
  \begin{itemize}
	\item $\mathcal{A}$ makes at most $Q$ evaluations of $H$, and
	\item the overhead of $\mathcal{A}$ is at most $I$.
  \end{itemize}
  Then \zcref{algo:blocks}, called with $j = 0$ on $\mathcal{I}$, returns the numbers $t^-_v$ of $\mathcal{I}$ in time $(Q + 1)^M \cdot (2^d + I) \cdot n^{O(1)}$.
  Also, each cap $N$ used in the recursion is at most~$2^{4n^2}$.
\end{lemma}
\begin{proof}
  We first bound the numbers of tokens that occur, which the hypothesis on $\mathcal{A}$ needs.
  For $j \in [0, M]$, let $\tau_j$ be the largest number of tokens of an instance at level $j$, so that $\tau_0 = t^+$.
  Consider a call at level $j < M$ whose instance has $\tau \le \tau_j$ tokens.
  It guesses the pivots of $B_{j + 1}$, and every guess is at most $N = 2^{n + 1} \cdot \tau$.
  So, as $|B_{j + 1}| \le k$,  every call that it makes has at most $\tau + k \cdot 2^{n + 1} \cdot \tau$ tokens.
  Since $k \le n$, this gives the recursive inequality
  \begin{equation*}
	\tau_{j + 1} \le (1 + n \cdot 2^{n + 1}) \cdot \tau_j \le 2^{2n + 1} \cdot \tau_j\enspace.
  \end{equation*}
  Hence $\tau_j \le 2^{(2n + 1)j} \cdot t^+$ for every $j \in [0, M]$.
  We have $M \le |S| \le n$ and $t^+ \le 2^n$, so every instance of the recursion has at most
  \begin{equation*}
	2^{(2n + 1)n} \cdot 2^n = 2^{2n^2 + 2n} \le 2^{3n^2}
  \end{equation*}
  tokens, using $n \ge 2$.
  So the cap of every call satisfies $N \le 2^{n + 1} \cdot 2^{3n^2} \le 2^{4n^2}$.

  Now we show by induction on $M - j$ that a call at level $j$ returns the numbers of tokens that arrive at the terminals of its instance $\mathcal{I}_j$.
  For $j = M$, the set $T^{(M)} = T \cup S$ is distance-$d$ dominating by \zcref{lem:pivots}, so \zcref{prop:base} applies.
  For $j < M$, note that $B_{j + 1}$ is a set of at most $k$ non-terminals of $\mathcal{I}_j$.
  Every evaluation of $H$ at a guess $\mathbf{w} \in [0, N]^{B_{j + 1}}$ is a call at level $j + 1$ on $\mathcal{I}_j(B_{j + 1} \leftarrow \mathbf{w})$, whose terminal set is $T^{(j + 1)}$.
  By induction, this call returns the numbers of arriving tokens of this instance, and their entries at the pivots form $\gmap_{B_{j + 1}}(\mathbf{w})$.
  So line~\ref{line:alg2-read} computes $H(\mathbf{w}) = \min\{\gmap_{B_{j + 1}}(\mathbf{w}), N\1\}$, the capped guessing map of \zcref{lem:capped} for~$\mathcal{I}_j$ and $B_{j + 1}$.
  By \zcref{lem:capped}, $H$ has unit increments, so $\mathcal{A}$ returns the least fixed point~$\mathbf{w}^*$ of~$H$, which is a consistent guess.
  By induction, line~\ref{line:alg2-final} computes the numbers of arriving tokens of $\mathcal{I}_j(B_{j + 1} \leftarrow \mathbf{w}^*)$.
  By \zcref{lem:consistent}, they agree with those of $\mathcal{I}_j$ at the terminals $T^{(j)}$ of $\mathcal{I}_j$.

  All numbers that occur have $\mathcal{O}(n^2)$ bits, since the numbers of tokens is at most~$2^{3n^2}$ and the guesses are at most $N \le 2^{4n^2}$.
  Hence there is a polynomial $P \ge 1$ in $n$ with the following two properties.
  Besides the recursive calls, a call at level $j < M$ spends time at most $P$ on line~\ref{line:alg2-cap}, at most $P$ per evaluation on lines~\ref{line:alg2-call} and~\ref{line:alg2-read}, and at most $P$ on lines~\ref{line:alg2-final} and~\ref{line:alg2-return}.
  A call at level $M$ takes time at most $2^d \cdot P$, by \zcref{prop:base}.

  We now analyse the run time. A call at level $j < M$ makes at most $Q + 1$ calls at level $j + 1$, one per evaluation of $H$ and one in line~\ref{line:alg2-final}.
  Besides these calls, it spends time at most $I + (Q + 2)P$.
  This is the overhead $I$ of $\mathcal{A}$ plus the time of the lines above.
  So there are at most $(Q + 1)^j$ calls at level $j$, and the total time is at most
  \begin{align*}
	& \sum_{j = 0}^{M - 1}(Q + 1)^j\bigl(I + (Q + 2)P\bigr) + (Q + 1)^M \cdot 2^d \cdot P\\
	& \le M \cdot (Q + 1)^M \cdot (I + 2P) + (Q + 1)^M \cdot 2^d \cdot P\\
	& \le (Q + 1)^M \cdot (2^d + I) \cdot (2M + 1)P\enspace,
  \end{align*}
  where $(2M + 1)P \le (2n + 1)P$ is again polynomial in $n$.
\end{proof}

\section{Searching a block by counting tokens}
\label{sec:counting}

In this section, we describe the fixed-point search used by the block solver. Each round tests a guess and removes an $\ell_1$-half-space from the remaining candidates. We show that this removal can include a fixed point when tokens circulate among the pivots (\zcref{sec:half-spaces}). We then introduce a leaky map that lets the search identify the least fixed point of the original map even when this happens (\zcref{sec:sinks}). Finally, we give the cutting procedure and prove that it returns the least fixed point whenever it stops (\zcref{sec:cutting}).

Throughout \zcref{sec:counting,sec:deep}, $k \ge 1$ and $N \ge 0$ are integers, and $F \colon [0, N]^k \to [0, N]^k$ is a map with unit increments (\zcref{def:unit-increments}).
When \zcref{algo:blocks} calls the solver at level $j$, the $k$ used here is $|B_{j+1}|$, and $F$ is the capped guessing map $H$ for that block. We prove the properties needed for the search using only unit increments, so \zcref{sec:sinks} can also use them for its leaky map and the corresponding map without leaks on a larger box.
We keep the language of tokens, in which $w_i$ is the number of tokens that pivot $i$ sends and $F(\mathbf{w})_i$ is the number that arrive at pivot $i$.
If $F(\mathbf{w} + \e_i) - F(\mathbf{w}) = \e_j$, then the additional token that pivot $i$ sends \emph{returns at} pivot~$j$.
If $F(\mathbf{w} + \e_i) = F(\mathbf{w})$, then it is \emph{absorbed}.
For a capped guessing map, a token is absorbed if it arrives at a terminal of $T$ or at a pivot whose entry is capped.
As in \zcref{sec:approach}, the \emph{surplus} of a guess $\mathbf{w}$ is $\surplus(\mathbf{w}) \coloneqq F(\mathbf{w}) - \mathbf{w}$, so $\mathbf{w}$ is a fixed point of $F$ if and only if $\surplus(\mathbf{w}) = \0$.

Unit increments bound how much $F$ changes between two guesses.
Every step of a path of guesses from~$\mathbf{x}$ to $\mathbf{y} \ge \mathbf{x}$ (the remark after \zcref{def:unit-increments}) adds one token, which returns or is absorbed.
Summing over the $\|\mathbf{y} - \mathbf{x}\|_1$ steps gives
\begin{equation}
\label{eq:no-multiply}
	\0 \le F(\mathbf{y}) - F(\mathbf{x})\quad\text{and}\quad\|F(\mathbf{y}) - F(\mathbf{x})\|_1 \le \|\mathbf{y} - \mathbf{x}\|_1\qquad\text{for all }\mathbf{x} \le \mathbf{y}\text{ in }[0, N]^k,
\end{equation}
where $\|F(\mathbf{y}) - F(\mathbf{x})\|_1$ counts the tokens of the path that return, and it does not depend on the path.
So equality holds on the right of~\eqref{eq:no-multiply} if and only if every token of every path from $\mathbf{x}$ to $\mathbf{y}$ returns. The norm inequality in~\eqref{eq:no-multiply} can be extended to arbitrary $\mathbf{x}, \mathbf{y} \in [0, N]^k$. Let $\mathbf{m}=\min\{\mathbf{x}, \mathbf{y}\}$ be their entrywise minimum. Then  the triangle inequality and~\eqref{eq:no-multiply} for $\mathbf{m} \le \mathbf{x}$ and $\mathbf{m} \le \mathbf{y}$
give
\begin{equation}
\label{eq:non-expansive}
	\|F(\mathbf{x}) - F(\mathbf{y})\|_1 \le \|\mathbf{x} - \mathbf{m}\|_1 + \|\mathbf{y} - \mathbf{m}\|_1 = \|\mathbf{x} - \mathbf{y}\|_1\qquad\text{for all }\mathbf{x}, \mathbf{y} \in [0, N]^k\enspace.
\end{equation}
\textcite[Lemma~12]{Haslebacher2025} shows the same inequality for the map that moves every token one step, which connects \garrival to the $\ell_1$-distance.
Since $F$ is monotone, it has a least fixed point $\mathbf{q}^*$ (\zcref{lem:kleene}).

\subsection{\texorpdfstring{$\ell_1$}{l1}-half-spaces}
\label{sec:half-spaces}
For a guess $\mathbf{w}$ with $\surplus(\mathbf{w}) \ne \0$, the \emph{wrong side} of $\mathbf{w}$ is the $\ell_1$-half-space $\halfspace{\mathbf{w}}{-\surplus(\mathbf{w})}$ (\zcref{def:half-space}).
We first describe $\ell_1$-half-spaces concretely.
For a point $\mathbf{c}$ and a direction $\mathbf{v}$, whether $\mathbf{z} \in \halfspace{\mathbf{c}}{\mathbf{v}}$ depends on $\mathbf{z}$ only through the sign vector $\sgn(\mathbf{z} - \mathbf{c}) \in \{-1, 0, 1\}^k$.
It records, at each pivot, whether $\mathbf{z}$ sends fewer, as many or more tokens than $\mathbf{c}$.
The next lemma follows from results of \textcite[Lemma~4.5 and Observation~4.6]{HaslebacherLSW2025}, and we include a short direct proof for the $\ell_1$-distance.
\zcref{fig:half-space} shows two examples.

\begin{lemma}
\label{lem:orthant}
  Let $\mathbf{z}, \mathbf{c} \in \mathbb{R}^k$, $\mathbf{v} \in \mathbb{R}^k \setminus \{\0\}$, and $\boldsymbol{\sigma} \coloneqq \sgn(\mathbf{z} - \mathbf{c})$.
  Then $\mathbf{z} \in \halfspace{\mathbf{c}}{\mathbf{v}}$ if and only if $\sum_{i : \sigma_i \ne 0}\sigma_iv_i + \sum_{i : \sigma_i = 0}|v_i| \ge 0$.
\end{lemma}
\begin{proof}
  Let $\varphi(\varepsilon) \coloneqq \|\mathbf{c} - \varepsilon\mathbf{v} - \mathbf{z}\|_1 = \sum_{i \in [k]}|z_i - c_i + \varepsilon v_i|$ be the distance from $\mathbf{z}$ to the moved point, which is a convex and piecewise linear function of $\varepsilon$.
  We show that the sum in the statement is the right derivative of $\varphi$ at $0$.
  By convexity, the sign of this derivative then decides whether $\mathbf{z} \in \halfspace{\mathbf{c}}{\mathbf{v}}$.
  For all small enough $\varepsilon > 0$, the $i^{\textnormal{th}}$ term equals $\sigma_i(z_i - c_i) + \varepsilon\sigma_iv_i$ if $\sigma_i \ne 0$, since $z_i - c_i + \varepsilon v_i$ then has the sign~$\sigma_i$.
  If $\sigma_i = 0$, it equals $\varepsilon|v_i|$.
  So the right derivative of $\varphi$ at $0$ is the sum in the statement.
  If this sum is non-negative, then $\varphi$ is non-decreasing for $\varepsilon \ge 0$ by convexity, so $\varphi(\varepsilon) \ge \varphi(0) = \|\mathbf{c} - \mathbf{z}\|_1$ for all $\varepsilon > 0$, and $\mathbf{z} \in \halfspace{\mathbf{c}}{\mathbf{v}}$.
  Otherwise, $\varphi(\varepsilon) < \varphi(0)$ for all small enough $\varepsilon > 0$, and $\mathbf{z} \notin \halfspace{\mathbf{c}}{\mathbf{v}}$.
\end{proof}

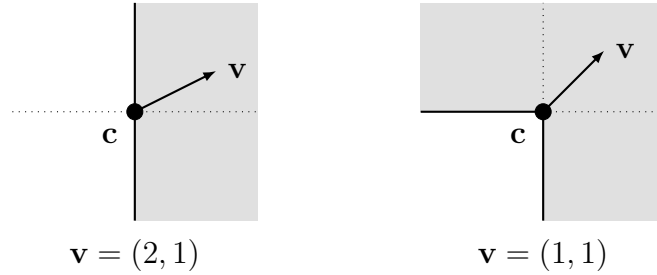
\begin{figure}[ht]
	\centering
	\begin{tikzpicture}[scale=0.9]
		\begin{scope}
		\fill[black!12] (0,-1.6) rectangle (1.8,1.6);
		\draw[thick] (0,-1.6) -- (0,1.6);
		\draw[dotted] (-1.8,0) -- (1.8,0);
		\draw[arrow, thick] (0,0) -- (1.2,0.6) node[right] {$\mathbf{v}$};
		\node[blackvertex, scale=1.2, label=below left:$\mathbf{c}$] at (0,0) {};
		\node at (0, -2.1) {$\mathbf{v} = (2, 1)$};
		\end{scope}
		\begin{scope}[xshift=6cm]
		\fill[black!12] (-1.8,0) rectangle (1.8,1.6);
		\fill[black!12] (0,-1.6) rectangle (1.8,0);
		\draw[thick] (-1.8,0) -- (0,0) -- (0,-1.6);
		\draw[dotted] (0,0) -- (1.8,0);
		\draw[dotted] (0,0) -- (0,1.6);
		\draw[arrow, thick] (0,0) -- (0.9,0.9) node[right] {$\mathbf{v}$};
		\node[blackvertex, scale=1.2, label=below left:$\mathbf{c}$] at (0,0) {};
		\node at (0, -2.1) {$\mathbf{v} = (1, 1)$};
		\end{scope}
	\end{tikzpicture}
	\caption{Two $\ell_1$-half-spaces $\halfspace{\mathbf{c}}{\mathbf{v}}$ in the plane (shaded, with their boundaries drawn solid where they belong to the set). On the left, $\mathbf{v} = (2, 1)$. By \zcref{lem:orthant}, a point $\mathbf{z}$ with $\boldsymbol{\sigma} = \sgn(\mathbf{z} - \mathbf{c})$ lies in $\halfspace{\mathbf{c}}{\mathbf{v}}$ if and only if $\sigma_1 \ge 0$, so $\halfspace{\mathbf{c}}{\mathbf{v}}$ is the closed half-plane to the right of $\mathbf{c}$. On the right, $\mathbf{v} = (1, 1)$. The sum of \zcref{lem:orthant} is negative only for $\boldsymbol{\sigma} = (-1, -1)$, so $\halfspace{\mathbf{c}}{\mathbf{v}}$ is the plane without the open quadrant below and to the left of $\mathbf{c}$.}
	\label{fig:half-space}
\end{figure}

The next lemma concerns the wrong side of a guess, and also any two fixed points.
For maps with unit increments, it is analoguous
to a result of \textcite[Lemma~4.1, Observation~3.4]{HaslebacherLSW2025}.
For maps that shrink all distances by a constant factor smaller than one, they show that the fixed point is strictly closer to $F(\mathbf{w})$ than to $\mathbf{w}$.

\begin{lemma}
\label{lem:tie}
  Let $\mathbf{q} \in [0, N]^k$ be a fixed point of $F$, and let $\mathbf{w} \in [0, N]^k$ be a fixed point of $F$ as well, or a guess with $\surplus(\mathbf{w}) \ne \0$ and $\mathbf{q} \in \halfspace{\mathbf{w}}{-\surplus(\mathbf{w})}$.
  Let $\mathbf{m}=\min\{\mathbf{w}, \mathbf{q}\}$.
  Then every step along every path of guesses from $\mathbf{m}$ to $\mathbf{w}$, and along every path of guesses from $\mathbf{m}$ to $\mathbf{q}$, causes a unit increase in one coordinate of $F$.
\end{lemma}
\begin{proof}
  We show that $\|\mathbf{w} - \mathbf{q}\|_1 \le \|F(\mathbf{w}) - F(\mathbf{m})\|_1 + \|F(\mathbf{q}) - F(\mathbf{m})\|_1$, which forces equality in~\eqref{eq:no-multiply} between $\mathbf{m}$ and $\mathbf{w}$ and between $\mathbf{m}$ and $\mathbf{q}$.
  If $\mathbf{w}$ is a fixed point, then $\|\mathbf{w} - \mathbf{q}\|_1 = \|F(\mathbf{w}) - F(\mathbf{q})\|_1$.
  Otherwise, $\mathbf{q} \in \halfspace{\mathbf{w}}{-\surplus(\mathbf{w})}$, and \zcref{def:half-space} with $\varepsilon = 1$ gives $\|\mathbf{w} - \mathbf{q}\|_1 \le \|\mathbf{w} + \surplus(\mathbf{w}) - \mathbf{q}\|_1 = \|F(\mathbf{w}) - F(\mathbf{q})\|_1$.
  In both cases, the triangle inequality and~\eqref{eq:no-multiply} give
  \begin{align*}
	\|\mathbf{w} - \mathbf{q}\|_1 &\le \|F(\mathbf{w}) - F(\mathbf{q})\|_1 \le \|F(\mathbf{w}) - F(\mathbf{m})\|_1 + \|F(\mathbf{q}) - F(\mathbf{m})\|_1\\
	&\le \|\mathbf{w} - \mathbf{m}\|_1 + \|\mathbf{q} - \mathbf{m}\|_1,
  \end{align*}
  and the right-hand side equals $\|\mathbf{w} - \mathbf{q}\|_1$, since $\mathbf{m}=\min\{\mathbf{w}, \mathbf{q}\}$.
  So equality holds throughout, in particular in~\eqref{eq:no-multiply} for the pairs $\mathbf{m} \le \mathbf{w}$ and $\mathbf{m} \le \mathbf{q}$.
  This means that every token of every path from $\mathbf{m}$ to $\mathbf{w}$ or to $\mathbf{q}$ returns, and the assertion follows.
\end{proof}

In the second case of \zcref{lem:tie}, 
$F(\mathbf{w})$ and $\mathbf{w}$ have the same sum of entries.
Indeed, every step of a path of guesses from $\mathbf{m}$ to $\mathbf{q}$ increases the sum of the entries of the guess by one.
By \zcref{lem:tie}, it also increases the sum of the entries of $F$
by one. Summing over the path gives $\sum_{i\in[k]}F(\mathbf{q})_i-\sum_{i\in[k]}F(\mathbf{m})_i
=\sum_{i\in[k]}q_i-\sum_{i\in[k]}m_i$.
Since $F(\mathbf{q})=\mathbf{q}$, we obtain
$\sum_{i\in[k]}F(\mathbf{m})_i=\sum_{i\in[k]}m_i$.
So $F(\mathbf{m})$ and $\mathbf{m}$ have the same sum of entries.
Every token of a path from $\mathbf{m}$ to $\mathbf{w}$ returns as well, so $\mathbf{w}$ has the same property.
So a fixed point lies on the wrong side of $\mathbf{w}$ only if $\mathbf{w}$ is a \emph{tie}, that is, a guess that is not a fixed point and  
$F(\mathbf{w})$ and $\mathbf{w}$ have the same sum of entries.
Ties do occur. The map $F(w_1, w_2) \coloneqq (w_2, w_1)$ on $[0, 1]^2$ has unit increments.
The guess $\mathbf{w} = (1, 0)$ has surplus $(-1, 1)$, both fixed points $(0, 0)$ and $(1, 1)$ lie in $\halfspace{\mathbf{w}}{-\surplus(\mathbf{w})}$, and the token that pivot~$1$ adds on the way from $(0, 0)$ to $\mathbf{w}$ returns at pivot~$2$.
The leaky map of \zcref{sec:sinks} makes ties harmless.
\zcref{fig:removed} shows a removed $\ell_1$-half-space at a guess that is not a tie.

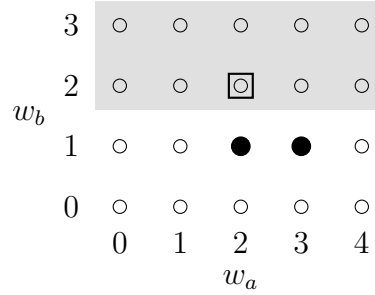
\begin{figure}[ht]
	\centering
	\begin{tikzpicture}[scale=.8]
		\fill[black!12] (-0.4,1.6) rectangle (4.4,3.4);
		\foreach \x in {0,...,4} {
			\foreach \y in {0,...,3} {
				\node[vertex] at (\x,\y) {};
			}
			\node at (\x,-0.6) {$\x$};
		}
		\foreach \y in {0,...,3} {
			\node at (-0.8,\y) {$\y$};
		}
		\node at (2,-1.2) {$w_a$};
		\node at (-1.5,1.5) {$w_b$};
		\node[blackvertex, scale=1.4] at (2,1) {};
		\node[blackvertex, scale=1.4] at (3,1) {};
		\node[draw, thick, minimum size=9pt, inner sep=0pt] at (2,2) {};
	\end{tikzpicture}
	\caption{The guesses $(w_a, w_b)$ with $w_a \le 4$ and $w_b \le 3$ for the pivots $a, b$ in the instance of \zcref{fig:fixed-point-example}. The fixed points of $\gmap_{\{a, b\}}$ in this range are $(2, 1)$ and $(3, 1)$ (black). The guess $\mathbf{w} = (2, 2)$ (square) gives $\gmap_{\{a, b\}}(\mathbf{w}) = (2, 1)$, so $\surplus(\mathbf{w}) = (0, -1)$. The pivot $a$ receives as many tokens as it sends, and the pivot $b$ receives one token less. By \zcref{lem:orthant}, $\halfspace{\mathbf{w}}{-\surplus(\mathbf{w})}$ consists of the guesses with $w_b \ge 2$ (shaded), and since the entries of $\surplus(\mathbf{w})$ do not sum to zero, none of them is consistent, by the remark after \zcref{lem:tie}.}
	\label{fig:removed}
\end{figure}

\subsection{The leaky map}
\label{sec:sinks}
If a fixed point lies on the wrong side of a tie, \zcref{lem:tie} says that every token added to reach either guess from their entrywise minimum returns to a pivot. The same holds for two fixed points. We now modify the map to absorb tokens at regular intervals, so that either kind of pair must be close in every coordinate.
This is similar in spirit to the discounting of \textcite[Section~4.2]{Haslebacher2025}.
In that technique, the map that moves every token one step is multiplied by a factor smaller than one, so that a fraction of the tokens disappears in every step and the map shrinks $\ell_1$-distances by this factor~\cite[Corollary~16]{Haslebacher2025}.
We assume $N \ge 1$ in the rest of \zcref{sec:counting}, since the case $N = 0$ is trivial.

\begin{definition}[Leaky map and rounding]
\label{def:leaky-map}
Let $\lambda \coloneqq 16(kN)^2$ and let $\Lambda \coloneqq \{0, \lambda/4, \lambda/2, 3\lambda/4\}$.
The \emph{leaky map} $F'$ and the \emph{rounding} $\rd$ are defined, for $\mathbf{y} \in [0, \lambda N]^k$, by
\begin{equation*}
	F'(\mathbf{y}) \coloneqq \sum_{a \in [0, \lambda - 1] \setminus \Lambda}F\Bigl(\Bigl\lfloor\frac{\mathbf{y} + a\1}{\lambda}\Bigr\rfloor\Bigr)\qquad\text{and}\qquad\rd(\mathbf{y}) \coloneqq \Bigl\lfloor\frac{\mathbf{y} + (\lambda/2)\1}{\lambda}\Bigr\rfloor\enspace.
\end{equation*}
\end{definition}

Imagine $\lambda$ copies of the set of pivots. For each pivot $i$, split its $y_i$ tokens as evenly as possible among the copies: copy $a$ gets $\lfloor(y_i+a)/\lambda\rfloor$ tokens at that pivot.
Apply~$F$ to the token counts of each copy.
The map $F'$ adds the arrivals from all copies except the four with indices in $\Lambda$; their arrivals are discarded, which creates the \emph{leak}.
The map $\rd$ turns $\mathbf{y}$ into a guess in the original box by rounding $\mathbf{y}/\lambda$. Finally, $\surplus'(\mathbf{y})\coloneqq F'(\mathbf{y})-\mathbf{y}$ records, at each pivot, the difference between the arrivals counted by $F'$ and the tokens in the guess.

Recall that $\mathbf{q}^*$ denotes the least fixed point of $F$.
\begin{lemma}
\label{lem:sinks}
The leaky map is a map $F' \colon [0, \lambda N]^k \to [0, \lambda N]^k$ with unit increments.
For every $\mathbf{y} \in [0, \lambda N]^k$, we have $\rd(\mathbf{y}) \in [0, N]^k$, and $F'(\mathbf{y})$ and $F(\rd(\mathbf{y}))$ can be computed with at most $k + 1$ evaluations of $F$.
Every fixed point $\mathbf{y}$ of $F'$ satisfies $\rd(\mathbf{y}) = \mathbf{q}^*$.
Moreover, if $\mathbf{y} \in [0, \lambda N]^k$ and $\rd(\mathbf{y}) \ne \mathbf{q}^*$, then $\surplus'(\mathbf{y}) \ne \0$, and $\halfspace{\mathbf{y}}{-\surplus'(\mathbf{y})}$ contains no fixed point of~$F'$.
\end{lemma}
\begin{proof}
For each $a\in[0,\lambda-1]$, let $\mathbf{g}_a(\mathbf{y})\coloneqq\lfloor(\mathbf{y}+a\1)/\lambda\rfloor$ be the vector of token counts assigned to copy $a$.
To see what the leak changes, first count the arrivals from \emph{all} copies: set $F_\lambda(\mathbf{y})\coloneqq\sum_{a=0}^{\lambda-1}F(\mathbf{g}_a(\mathbf{y}))$.
The leaky map uses the same token counts but omits the outputs of the four copies in $\Lambda$, so $F'(\mathbf{y})=F_\lambda(\mathbf{y})-\sum_{a\in\Lambda}F(\mathbf{g}_a(\mathbf{y}))$.
We use~$F_\lambda$ as a reference: on a scaled guess $\lambda\mathbf{x}$ it equals $\lambda F(\mathbf{x})$, allowing us to measure how much the omitted outputs change the fixed-point search.

We first show that $F'$ and $F_\lambda$ have unit increments and map $[0, \lambda N]^k$ to itself, and that~$F'$ absorbs one of any $\lambda/4$ consecutive tokens of a pivot.
For $\mathbf{y} \in [0, \lambda N]^k$ and every $a$, we have $0 \le \mathbf{g}_a(\mathbf{y})_i \le \lfloor(\lambda N + \lambda - 1)/\lambda\rfloor = N$, so $\mathbf{g}_a(\mathbf{y}) \in [0, N]^k$.
Consider a step from~$\mathbf{y}$ to $\mathbf{y} + \e_i$ in $[0, \lambda N]^k$.
The difference $\lfloor(y_i + 1 + a)/\lambda\rfloor - \lfloor(y_i + a)/\lambda\rfloor$ is one if~$\lambda$ divides $y_i + 1 + a$ and zero otherwise.
So exactly one copy receives the extra token at pivot~$i$, namely the copy $a^*$ with $a^* \equiv -(y_i + 1) \pmod\lambda$, and $\mathbf{g}_{a^*}(\mathbf{y} + \e_i) = \mathbf{g}_{a^*}(\mathbf{y}) + \e_i$.
Hence $F'(\mathbf{y} + \e_i) - F'(\mathbf{y})$ is $F(\mathbf{g}_{a^*}(\mathbf{y}) + \e_i) - F(\mathbf{g}_{a^*}(\mathbf{y})) \in \{\0, \e_1, \ldots, \e_k\}$ if $a^* \notin \Lambda$, and~$\0$ if $a^* \in \Lambda$.
The same holds for~$F_\lambda$ without the exception $a^* \in \Lambda$.
Also $\0 \le F' \le F_\lambda$, and $F_\lambda(\mathbf{y})_j \le \lambda N$ (as $\mathbf{g}_a(\mathbf{y})$ lies in $[0,  N]^k$ and $F$ maps this box to itself). So~$F'$ and~$F_\lambda$ map $[0, \lambda N]^k$ to itself.
As $y_i$ runs through~$\lambda/4$ consecutive values, $a^*$ runs through $\lambda/4$ cyclically consecutive residues modulo $\lambda$, one of which is a multiple of~$\lambda/4$, that is, in $\Lambda$.
In any $\lambda/4$ consecutive steps that increase the same coordinate, at least one leaves $F'$ unchanged because it reaches a leaky copy.

Next, we show the claims about $\rd$.
Every coordinate of $\mathbf{g}_a(\mathbf{y})$ is non-decreasing in $a$, and it increases at most once as $a$ runs from $0$ to $\lambda - 1$.
A change can occur only when $a=\lambda-(y_i\bmod\lambda)$ for some $i\in[k]$. These at most $k$ change points divide $[0,\lambda-1]$ into at most $k+1$ intervals, and $\mathbf g_a(\mathbf y)$ is constant on each of them. So it takes at most $k + 1$ values, and $\rd(\mathbf y)=\mathbf g_{\lambda/2}(\mathbf y)$ is one of them.
Hence $\rd(\mathbf{y}) \in [0, N]^k$, and $F'(\mathbf{y})$ and $F(\rd(\mathbf{y}))$ are computed with at most $k + 1$ evaluations of $F$.
To compute~$F'(\mathbf{y})$, evaluate $F$ once for each such interval $J$, multiply the result by $|J \setminus \Lambda|$, and add the results. Apart from these evaluations of $F$, this takes $\mathcal{O}(k^2)$ arithmetic operations.

Next, we show that the fixed points of $F'$ are close to each other.
We also show that a fixed point $\mathbf{q}$ is close to a guess $\mathbf{y}$ if $\mathbf{q}$ lies on the wrong side of $\mathbf{y}$.
Let $\mathbf{q}$ be a fixed point of $F'$, and let $\mathbf{y} \in [0, \lambda N]^k$ be a fixed point of $F'$ as well, or a guess with $\surplus'(\mathbf{y}) \ne \0$ and $\mathbf{q} \in \halfspace{\mathbf{y}}{-\surplus'(\mathbf{y})}$.
We show that $|q_i - y_i| \le \lambda/4 - 1$ for all $i \in [k]$.
Let $\mathbf{m}=\min\{\mathbf{y}, \mathbf{q}\}$,
and let $i \in [k]$ with $q_i > y_i$.
A path from $\mathbf{m}$ to $\mathbf{q}$ can begin with the $q_i - y_i$ steps in coordinate $i$, which are consecutive steps along one line parallel to an axis.
By \zcref{lem:tie}, applied to the map $F'$, which has unit increments on the box $[0, \lambda N]^k$, 
none of their tokens is absorbed by $F'$.
But, as shown above, every $\lambda/4$ consecutive steps in coordinate $i$ include an absorbed token.
Thus $q_i - y_i \le \lambda/4 - 1$.
Now let $i \in [k]$ with $q_i < y_i$.
A path from $\mathbf{m}$ to $\mathbf{y}$ can begin with the $y_i - q_i$ steps in coordinate $i$, and \zcref{lem:tie} says that none of their tokens is absorbed by $F'$ either.
So the same argument gives $y_i - q_i \le \lambda/4 - 1$.
In particular, fixed points $\mathbf{p} \le \mathbf{p}'$ of $F'$ satisfy $\mathbf{p}' \le \mathbf{p} + (\lambda/4 - 1)\1$.

Now we show that the fixed points of $F'$ are close to $\lambda\mathbf{q}^*$.
Let $\mathbf{u}_i \coloneqq F^i(\0)$, $\mathbf{u}^\lambda_i \coloneqq F_\lambda^i(\0)$ and $\mathbf{u}'_i \coloneqq F'^i(\0)$ be the iterations of \zcref{lem:kleene} for the three maps, so that $\mathbf{u}_m = \mathbf{q}^*$ for some $m \le kN$.
Since $\mathbf{g}_a(\lambda\mathbf{x}) = \mathbf{x}$ for all $a$ and all $\mathbf{x} \in [0, N]^k$, we have $F_\lambda(\lambda\mathbf{x}) = \lambda F(\mathbf{x})$.
By induction, $\mathbf{u}^\lambda_i = \lambda\mathbf{u}_i$, so $\mathbf{u}^\lambda_m = \lambda\mathbf{q}^*$.
Since $F' \le F_\lambda$ and $F_\lambda$ is monotone, $\mathbf{u}'_i \le \mathbf{u}^\lambda_i$ by induction.
The map $F_\lambda$ has unit increments, so~\eqref{eq:no-multiply} holds for it, and $\|F(\mathbf{u})\|_1 \le kN$ for $\mathbf{u} \in [0, N]^k$, so
\begin{align*}
	\|\mathbf{u}^\lambda_{i + 1} - \mathbf{u}'_{i + 1}\|_1
	& \le \|F_\lambda(\mathbf{u}^\lambda_i) - F_\lambda(\mathbf{u}'_i)\|_1 + \|F_\lambda(\mathbf{u}'_i) - F'(\mathbf{u}'_i)\|_1\\
	& \le \|\mathbf{u}^\lambda_i - \mathbf{u}'_i\|_1 + \sum_{a \in \Lambda}\|F(\mathbf{g}_a(\mathbf{u}'_i))\|_1 \\
	&\le \|\mathbf{u}^\lambda_i - \mathbf{u}'_i\|_1 + 4kN\enspace.
\end{align*}
Hence $\0 \le \lambda\mathbf{q}^* - \mathbf{u}'_m$ and $\|\lambda\mathbf{q}^* - \mathbf{u}'_m\|_1 \le 4kN \cdot m \le 4(kN)^2 = \lambda/4$.
Moreover, $F'(\lambda\mathbf{q}^*) \le F_\lambda(\lambda\mathbf{q}^*) = \lambda\mathbf{q}^*$, so $\mathbf{u}'_i \le \lambda\mathbf{q}^*$ for all $i$, by induction and monotonicity of $F'$.
By \zcref{lem:kleene}, the non-decreasing sequence $(\mathbf{u}'_i)$ becomes stationary at the least fixed point $\mathbf{q}'$ of $F'$, so $\mathbf{u}'_m \le \mathbf{q}' \le \lambda\mathbf{q}^*$.
Every fixed point $\mathbf{y}$ of $F'$ satisfies $\mathbf{y} \ge \mathbf{q}' \ge \mathbf{u}'_m \ge \lambda\mathbf{q}^* - (\lambda/4)\1$.
By the previous paragraph applied to $\mathbf{q}' \le \mathbf{y}$, also $\mathbf{y} \le \mathbf{q}' + (\lambda/4 - 1)\1 \le \lambda\mathbf{q}^* + (\lambda/4 - 1)\1$.
So $\lambda q^*_i + \lambda/4\le y_i + \lambda/2\le\lambda q^*_i + 3\lambda/4 - 1$ for every $i$. 
Dividing by $\lambda$ and taking the floor gives
$\rd(\mathbf y)_i=\lfloor(y_i+\lambda/2)/\lambda\rfloor=q_i^*$
for every $i$, and thus $\rd(\mathbf{y}) = \mathbf{q}^*$.

Finally, let $\mathbf{y} \in [0, \lambda N]^k$ be such that $\rd(\mathbf{y}) \ne \mathbf{q}^*$.
Then $\mathbf{y}$ is not a fixed point of $F'$, so $\surplus'(\mathbf{y}) \ne \0$.
Moreover, $y_i + \lambda/2$ lies outside $\{\lambda q^*_i, \ldots, \lambda q^*_i + \lambda - 1\}$ for some $i$, as otherwise $\rd(\mathbf{y}) = \mathbf{q}^*$.
By the bounds of the previous paragraph, $|y_i - q_i| \ge \lambda/4 + 1$ for every fixed point~$\mathbf{q}$ of $F'$.
So no fixed point $\mathbf{q}$ of $F'$ lies in $\halfspace{\mathbf{y}}{-\surplus'(\mathbf{y})}$, since otherwise $|q_i - y_i| \le \lambda/4 - 1$ for all $i$, as shown above.
\end{proof}

\subsection{Cutting}
\label{sec:cutting}
We now describe an algorithm for finding the least fixed point $\mathbf{q}^*$  of $F$.
It evaluates the leaky map $F'$ at successive guesses and, when a guess
is not a fixed point of $F'$, removes an $\ell_1$-half-space from the
remaining candidates. More precisely, the \emph{cutting algorithm} keeps a set $C \subseteq [0, \lambda N]^k$ of \emph{candidates}, initially $C = [0, \lambda N]^k$, and it repeats the following \emph{round} as long as $C$ is non-empty.
It chooses a guess $\mathbf{w} \in \bbox(C) \cap \mathbb{Z}^k$, and it computes $F'(\mathbf{w})$, which gives the surplus $\surplus'(\mathbf{w})$, and $F(\rd(\mathbf{w}))$.
If $\rd(\mathbf{w})$ is a fixed point of $F$, it records $\rd(\mathbf{w})$.
If $\surplus'(\mathbf{w}) = \0$, it outputs $\rd(\mathbf{w})$, which is $\mathbf{q}^*$ by \zcref{lem:sinks}.
Otherwise, it removes the $\ell_1$-half-space $\halfspace{\mathbf{w}}{-\surplus'(\mathbf{w})}$ from~$C$.
If~$C$ becomes empty, it outputs the recorded fixed point of $F$ with the smallest sum of entries.
A fixed point $\rd(\mathbf{w})$ of $F$ is output at once only if $\surplus'(\mathbf{w}) = \0$, since otherwise it need not be the least one.
A point $\mathbf{w}$ is \emph{$\alpha$-deep} in a finite set $C \subset \mathbb{R}^k$ if $|C \cap \halfspace{\mathbf{w}}{\mathbf{v}}| \ge \alpha|C|$ for all $\mathbf{v} \in \mathbb{R}^k \setminus \{\0\}$.
With deep guesses, this is a variant of the search for approximate fixed points of maps that shrink distances by a constant factor~\cite[Section~4]{HaslebacherLSW2025}, to which \garrival can be reduced for the $\ell_1$-distance.
We count the evaluations as in their proof~\cite[Theorem~4.3]{HaslebacherLSW2025}.

\begin{lemma}
\label{lem:cutting}
  Regardless of how each guess $\mathbf{w}$ is chosen, the following hold.
\begin{itemize}
  \item If the search stops, then it outputs the least fixed point $\mathbf{q}^*$ of $F$.
  \item After $r$ rounds, $C$ is the union of at most $(2r + 1)^k$ disjoint boxes, and the $r^{\textnormal{th}}$ round maintains this representation with $\mathcal{O}(k(2r + 1)^k)$ arithmetic operations on integers with $\mathcal{O}(\log(kN + 2))$ bits.
  \item If, for some real $\alpha$ with $0 < \alpha \le 1$, every guess $\mathbf{w}$ is $\alpha$-deep in the current set $C$, then the search stops after $\mathcal{O}(k \cdot \log(kN + 2)/\alpha)$ rounds and $\mathcal{O}(k^2 \cdot \log(kN + 2)/\alpha)$ evaluations of~$F$.
\end{itemize}
\end{lemma}
\begin{proof}
  We first show that the output is $\mathbf{q}^*$ if the search stops.
  Every guess lies in $\bbox(C) \cap \mathbb{Z}^k \subseteq [0, \lambda N]^k$.
  If the search stops at a guess $\mathbf{w}$ with $\surplus'(\mathbf{w}) = \0$, then $\mathbf{w}$ is a fixed point of $F'$, so its output $\rd(\mathbf{w})$ is $\mathbf{q}^*$ by \zcref{lem:sinks}.
  Otherwise, the search stops because $C$ has become empty.
  The map $F' \colon [0, \lambda N]^k \to [0, \lambda N]^k$ has unit increments, by \zcref{lem:sinks}, so it is monotone and has a fixed point by \zcref{lem:kleene}.
  This fixed point belongs to the initial set $C = [0, \lambda N]^k$, so some round removes a fixed point of $F'$.
  In this round, $\surplus'(\mathbf{w}) \ne \0$ and $\halfspace{\mathbf{w}}{-\surplus'(\mathbf{w})}$ contains a fixed point of $F'$, so $\rd(\mathbf{w}) = \mathbf{q}^*$ by \zcref{lem:sinks}.
  The round has recorded $\mathbf{q}^*$, since $F(\mathbf{q}^*) = \mathbf{q}^*$.
  Any recorded point is a fixed point of $F$, and so at least~$\mathbf{q}^*$ in every entry by \zcref{lem:kleene}.
  So $\mathbf{q}^*$ is the recorded point with the smallest sum of entries, and it is the output.

  For the second claim, suppose that, in each coordinate, $[0,\lambda N]$ is partitioned into at most $m$ intervals of consecutive integers. The products of one interval from each coordinate partition form a \emph{grid}, and we call these products its \emph{cells}.
  Initially, $C = [0, \lambda N]^k$ is the single cell of the grid with the intervals $[0, \lambda N]$.
  A round refines the grid by cutting, in every coordinate $i$, the interval that contains $w_i$ into the parts below $w_i$, at~$w_i$ and above $w_i$, which adds at most two intervals per coordinate.
  By \zcref{lem:orthant}, whether a point $\mathbf{z}$ lies in the removed set depends only on $\sgn(\mathbf{z} - \mathbf{w})$, which is constant on each cell of the refined grid, so the removed set is a union of such cells.
  So after $r$ rounds, $C$ is a union of cells of a grid with at most $2r + 1$ intervals per coordinate, and the cells are disjoint boxes.
  The $r^{\textnormal{th}}$ round refines the cells and tests, for each cell of the refined grid, if it lies in $\halfspace{\mathbf{w}}{-\surplus'(\mathbf{w})}$.
  The test applies \zcref{lem:orthant} to the sign vector of the cell relative to~$\mathbf{w}$, which takes $\mathcal{O}(k)$ operations per cell.
  All coordinates and all surpluses are integers of absolute value at most $\lambda N = 16k^2N^3$.

  For the third claim, consider a round after which the search continues.
  Then $\surplus'(\mathbf{w}) \ne \0$, and as $\mathbf{w}$ is $\alpha$-deep in $C$, the round removes at least $\alpha|C|$ candidates.
  A round starts only if $|C| \ge 1$, so if $r$ such rounds are followed by another one, then
  \begin{equation*}
    1 \le (1 - \alpha)^r \cdot (\lambda N + 1)^k \le e^{-\alpha r} \cdot (\lambda N + 1)^k,
  \end{equation*}
  as $C$ starts with $(\lambda N + 1)^k$ candidates.
  So $r \le k \cdot \ln(\lambda N + 1)/\alpha$, where $k \cdot \ln(\lambda N + 1) = k \cdot \ln(16k^2N^3 + 1) = \mathcal{O}(k \cdot \log(kN + 2))$.
  So the search stops after $1 + \mathcal{O}(k \cdot \log(kN + 2)/\alpha)$ rounds, and any round uses at most $k + 1$ evaluations of~$F$, by \zcref{lem:sinks}.
\end{proof}

\section{Deep guesses}
\label{sec:deep}
In every round of the cutting algorithm, we deterministically compute a guess that is $\frac1{2(k + 1)}$\nobreakdash-deep in the current candidate set $C$ (\zcref{lem:deep-point}).
\zcref{sec:centerpoints} shows that deep points exist and can be rounded down, \zcref{sec:exact} computes a deepest point of few weighted points by brute force, \zcref{sec:compression} computes a small weighted sample that represents the candidates, and \zcref{sec:solver} combines these steps.

\subsection{\texorpdfstring{$\ell_1$}{l1}-centerpoints}
\label{sec:centerpoints}
We use the $\ell_1$-centerpoint theorem of \textcite[Theorem~3.18]{HaslebacherLSW2025}.
It is stated for a family of distances indexed by a parameter $p$, with $p = 1$ for the $\ell_1$-distance, and with half-spaces $\mathcal{H}^p_{c, v}$~\cite[Definition~3.2]{HaslebacherLSW2025}, which for $p = 1$ are the $\ell_1$-half-spaces of \zcref{def:half-space}.
Their $d$ is the dimension, which is $k$ for us, and the directions~$v$ range over the set $S^{d - 1}$ of vectors of Euclidean norm one.
The statement reads ``Let $p \in [1, \infty) \cup \{\infty\}$ be arbitrary, and let $P \subseteq \mathbb{R}^d$ be a finite set of points.
There exists a point~$c$ such that $|\mathcal{H}^p_{c, v} \cap P| \ge \frac{|P|}{d + 1}$ for all $v \in S^{d - 1}$.''

Since $\halfspace{\mathbf{c}}{\mathbf{v}}$ does not change when $\mathbf{v}$ is multiplied by a positive number (the remark after \zcref{def:half-space}), the directions of Euclidean norm one cover all directions $\mathbf{v} \ne \0$.
For $p = 1$ and $d = k$, the theorem reads as follows.

\begin{proposition}[{\cite[Theorem~3.18]{HaslebacherLSW2025}}, for the $\ell_1$-distance]
\label{thm:hlsw}
  Let $X \subseteq \mathbb{R}^k$ be a finite set of points.
  There is a point $\mathbf{c} \in \mathbb{R}^k$ so that $|\halfspace{\mathbf{c}}{\mathbf{v}} \cap X| \ge \frac{|X|}{k + 1}$ for all $\mathbf{v} \in \mathbb{R}^k \setminus \{\0\}$.
\end{proposition}

We work with finite sets $X \subset \mathbb{R}^k$ with positive integer weights, since the compression of \zcref{sec:compression} produces weighted points, and with multisets, where the weight of a point is its multiplicity.
For such a set, $|X \cap A|$ denotes the total weight of the points of $X$ in a set $A$, and $|X|$ the total weight.
The \emph{depth} of a point $\mathbf{c} \in \mathbb{R}^k$ in $X$ is $\dep_X(\mathbf{c}) \coloneqq \min_{\mathbf{v} \ne \0}|X \cap \halfspace{\mathbf{c}}{\mathbf{v}}|$, so a point $\mathbf{w}$ is $\alpha$-deep in a finite set $C$ (\zcref{sec:cutting}) if and only if $\dep_C(\mathbf{w}) \ge \alpha|C|$.

The guesses must be integer points, and rounding down does not hurt, by the following lemma.
Its sign argument also appears in an earlier proof that rounds to the nearest point of a grid~\cite[Lemma~4.4]{HaslebacherLSW2025}.

\begin{lemma}
\label{lem:rounding}
  Let $\mathbf{c}, \mathbf{z} \in \mathbb{R}^k$ and $\mathbf{v} \in \mathbb{R}^k \setminus \{\0\}$.
  If $\mathbf{z} \in \halfspace{\mathbf{c}}{\mathbf{v}}$, then $\lfloor\mathbf{z}\rfloor \in \halfspace{\lfloor\mathbf{c}\rfloor}{\mathbf{v}}$.
  In particular, $\dep_C(\lfloor\mathbf{c}\rfloor) \ge \dep_C(\mathbf{c})$ for every finite weighted set $C \subset \mathbb{Z}^k$.
  Moreover, if $C$ is non-empty and $\mathbf{c} \in \bbox(C)$, then $\lfloor\mathbf{c}\rfloor \in \bbox(C)$.
\end{lemma}
\begin{proof}
  We compare $\boldsymbol{\sigma} \coloneqq \sgn(\mathbf{z} - \mathbf{c})$ with $\boldsymbol{\sigma}' \coloneqq \sgn(\lfloor\mathbf{z}\rfloor - \lfloor\mathbf{c}\rfloor)$.
  Rounding down is monotone, so $z_i > c_i$ gives $\lfloor z_i\rfloor \ge \lfloor c_i\rfloor$, $z_i < c_i$ gives $\lfloor z_i\rfloor \le \lfloor c_i\rfloor$, and $z_i = c_i$ gives $\lfloor z_i\rfloor = \lfloor c_i\rfloor$.
  So~$\boldsymbol{\sigma}'$ arises from $\boldsymbol{\sigma}$ by replacing some entries by $0$, which replaces summands $\sigma_iv_i \le |v_i|$ by $|v_i|$ in the sum of \zcref{lem:orthant} and can only increase the sum.
  Hence $\mathbf{z} \in \halfspace{\mathbf{c}}{\mathbf{v}}$ implies $\lfloor\mathbf{z}\rfloor \in \halfspace{\lfloor\mathbf{c}\rfloor}{\mathbf{v}}$.
  Every $\mathbf{z} \in C \subset \mathbb{Z}^k$ has $\lfloor\mathbf{z}\rfloor = \mathbf{z}$, so $C \cap \halfspace{\mathbf{c}}{\mathbf{v}} \subseteq C \cap \halfspace{\lfloor\mathbf{c}\rfloor}{\mathbf{v}}$ for every $\mathbf{v} \ne \0$, which gives the claim about the depths.
  The last claim holds since $\bbox(C)$ has integer corners.
\end{proof}

\zcref{thm:hlsw} extends to weighted sets as follows.
The first claim of the next lemma is known for the points that \zcref{thm:hlsw} gives~\cite[Lemma~3.19]{HaslebacherLSW2025}, and a version of the second claim for unweighted sets of grid points~\cite[Lemma~4.4]{HaslebacherLSW2025}.

\begin{lemma}
\label{lem:weighted-centerpoint}
  Let $X \subset \mathbb{Z}^k$ be a finite non-empty set with positive integer weights.
  Then each point of positive depth in $X$ lies in $\bbox(X)$, and there is a point $\mathbf{c} \in \bbox(X)$ with $\dep_X(\mathbf{c}) \ge \frac{|X|}{k + 1}$.
\end{lemma}
\begin{proof}
  For the first claim, let $\mathbf{c} \notin \bbox(X)$, so that $c_i < z_i$ for all $\mathbf{z} \in X$ or $c_i > z_i$ for all $\mathbf{z} \in X$, for some coordinate $i$.
  In the first case, every $\mathbf{z} \in X$ has $\sgn(\mathbf{z} - \mathbf{c})_i = 1$, so the sum of \zcref{lem:orthant} for $\mathbf{v} = -\e_i$ is $-1$, and no point of $X$ lies in $\halfspace{\mathbf{c}}{-\e_i}$.
  In the second case, no point of $X$ lies in $\halfspace{\mathbf{c}}{\e_i}$.
  So $\dep_X(\mathbf{c}) = 0$.

  For the second claim, we spread the weight of every point over distinct points nearby, apply \zcref{thm:hlsw}, and round the point that it gives down.
  Let $\mu_{\mathbf{y}}$ be the weight of $\mathbf{y} \in X$, let $\mu \coloneqq \max_{\mathbf{y} \in X}\mu_{\mathbf{y}}$, and let $\tilde X$ be the set of the points $\mathbf{y} + \frac j\mu\1$ with $\mathbf{y} \in X$ and $0 \le j < \mu_{\mathbf{y}}$.
  These are $|X|$ distinct points, since the points of $X$ are integer points and all offsets satisfy $0 \le \frac j\mu < 1$.
  By \zcref{thm:hlsw}, there is a point $\mathbf{c}' \in \mathbb{R}^k$ with $|\halfspace{\mathbf{c}'}{\mathbf{v}} \cap \tilde X| \ge \frac{|X|}{k + 1}$ for all $\mathbf{v} \ne \0$, and we let $\mathbf{c} \coloneqq \lfloor\mathbf{c}'\rfloor$.
  Every point $\mathbf{z} = \mathbf{y} + \frac j\mu\1$ of $\tilde X$ satisfies $\lfloor\mathbf{z}\rfloor = \mathbf{y}$, so $\mathbf{z} \in \halfspace{\mathbf{c}'}{\mathbf{v}}$ implies $\mathbf{y} \in \halfspace{\mathbf{c}}{\mathbf{v}}$ by \zcref{lem:rounding}.
  Since every $\mathbf{y} \in X$ gives $\mu_{\mathbf{y}}$ points of $\tilde X$, we get $|X \cap \halfspace{\mathbf{c}}{\mathbf{v}}| \ge |\tilde X \cap \halfspace{\mathbf{c}'}{\mathbf{v}}| \ge \frac{|X|}{k + 1}$ for all $\mathbf{v} \ne \0$.
  Finally, $\mathbf{c} \in \bbox(X)$ by the first claim, since $\frac{|X|}{k + 1} > 0$.
\end{proof}

\subsection{The case of few points}
\label{sec:exact}
For a weighted set of few points, a deepest point can be computed exactly, since the $\ell_1$-half-spaces cut out only few different subsets of the points, and those can be listed.
\begin{definition}[Traces]
\label{def:traces}
  For a multiset $Y = (\mathbf{y}_1, \ldots, \mathbf{y}_m)$ of points of $\mathbb{R}^k$, the \emph{traces} of the $\ell_1$-half-spaces on $Y$ are the index sets $\{j \in [m] : \mathbf{y}_j \in \halfspace{\mathbf{c}}{\mathbf{v}}\}$ for $\mathbf{c} \in \mathbb{R}^k$ and~\mbox{$\mathbf{v} \ne \0$}.
\end{definition}

In other words, a trace records which points of $Y$ lie in one $\ell_1$-half-space.
In our context, the points are candidates, and the candidates that one round of the cutting algorithm removes form a trace.
Many $\ell_1$-half-spaces have the same trace, and \zcref{lem:traces} shows that there are only few distinct traces, with the following classical bound.
A family $\mathcal{R}$ of subsets of a finite set \emph{shatters} a set $A$ if the sets $R \cap A$ with $R \in \mathcal{R}$ are all the subsets of $A$.

\begin{proposition}[Sauer--Shelah lemma~\cite{Sauer1972,Shelah1972}]
\label{prop:sauer}
  Let $m, s \in \N_0$.
  A family of subsets of an $m$-element set that shatters no $(s + 1)$-element set has at most $\sum_{i \le s}\binom mi$ members.
\end{proposition}

\begin{lemma}
\label{lem:traces}
  Let $Y = (\mathbf{y}_1, \ldots, \mathbf{y}_m)$ be a multiset of $m \ge 1$ points of $[0, 2^\beta]^k$, where $\beta \in \N$.
  One can compute deterministically, in time $(m + 1)^{\mathcal{O}(k)} \cdot \poly(\beta)$, a set $Z \subset (\frac12\mathbb{Z})^k$ and, for any $\mathbf{c} \in Z$, a family $\mathcal{U}_{\mathbf{c}}$ of subsets of $[m]$ for which it holds:
\begin{itemize}
	\item For every $\mathbf{c}' \in \mathbb{R}^k$, some $\mathbf{c} \in Z$ has $\sgn(\mathbf{y}_j - \mathbf{c}) = \sgn(\mathbf{y}_j - \mathbf{c}')$ for all $j \in [m]$.
	\item For every $\mathbf{c} \in Z$, the family $\mathcal{U}_{\mathbf{c}}$ consists of $[m]$ and of the traces on $Y$ of the $\ell_1$-half-spaces through $\mathbf{c}$.
	\item Every member of these families is a trace on $Y$, every trace on $Y$ is a member of one of them, and they have at most $(12m^2)^k$ members in total.
\end{itemize}
\end{lemma}
\begin{proof}
We first choose $Z$.
For $i \in [k]$, let $\xi_1 < \cdots < \xi_{p_i}$ be the distinct $i$-th coordinates of the points of $Y$.
The \emph{type} of a point $\mathbf{c}' \in \mathbb{R}^k$ records, for every $i$, the position of~$c'_i$ among these values, that is, whether $c'_i$ is equal to some $\xi_a$, strictly between $\xi_a$ and $\xi_{a + 1}$, below $\xi_1$ or above $\xi_{p_i}$.
The sign vectors $\sgn(\mathbf{y}_j - \mathbf{c}')$ depend only on the type of $\mathbf{c}'$, and there are at most $(2m + 1)^k$ types.
For every type, $Z$ contains one representative $\mathbf{c}$, with $c_i = \xi_a$, $c_i = (\xi_a + \xi_{a + 1})/2$, $c_i = \xi_1 - 1$ or $c_i = \xi_{p_i} + 1$, so that $\mathbf{c} \in (\frac12\mathbb{Z})^k$.
By \zcref{lem:orthant}, every $\ell_1$-half-space $\halfspace{\mathbf{c}'}{\mathbf{v}}$ has the same trace as $\halfspace{\mathbf{c}}{\mathbf{v}}$ for the representative $\mathbf{c}$ of the type of $\mathbf{c}'$, so the families $\mathcal{U}_{\mathbf{c}}$ contain all traces on $Y$.
The set $[m]$ is a trace as well.
By \zcref{lem:orthant}, the $\ell_1$-half-space through a point below all points of $Y$ in every coordinate, in the direction $\1$, contains all of them.

Next, we describe $\mathcal{U}_{\mathbf{c}}$ for $\mathbf{c} \in Z$ by linear inequalities, one orthant of directions at a time.
For $\boldsymbol{\varsigma} \in \{-1, 1\}^k$, consider the directions $\mathbf{v}$ in the closed orthant $O_{\boldsymbol{\varsigma}} \coloneqq \{\mathbf{v} : \varsigma_iv_i \ge 0\text{ for all }i\}$, where $|v_i| = \varsigma_iv_i$.
For $j \in [m]$, let $\mathbf{a}_j \in \{-1, 1\}^k$ have the entries $a_{j, i} \coloneqq \sgn(\mathbf{y}_j - \mathbf{c})_i$ if this is non-zero, and $a_{j, i} \coloneqq \varsigma_i$ otherwise.
Then the sum of \zcref{lem:orthant} for $\mathbf{y}_j$ and $\mathbf{v} \in O_{\boldsymbol{\varsigma}}$ is $\mathbf{a}_j \cdot \mathbf{v}$, so the trace of $\halfspace{\mathbf{c}}{\mathbf{v}}$ is $\{j : \mathbf{a}_j \cdot \mathbf{v} \ge 0\}$ for every $\mathbf{v} \in O_{\boldsymbol{\varsigma}} \setminus \{\0\}$, while $\mathbf{v} = \0$ gives $[m]$.
Since every $\mathbf{v} \ne \0$ lies in some orthant $O_{\boldsymbol{\varsigma}}$, the family $\mathcal{U}_{\mathbf{c}}$ is the union, over all $\boldsymbol{\varsigma}$, of the families~$\mathcal{U}_{\mathbf{c}, \boldsymbol{\varsigma}}$ of the sets $\{j : \mathbf{a}_j \cdot \mathbf{v} \ge 0\}$ with $\mathbf{v} \in O_{\boldsymbol{\varsigma}}$.

These families are small.
The members of $\mathcal{U}_{\mathbf{c}, \boldsymbol{\varsigma}}$ are subsets of $[m]$ that closed half-spaces $\{\mathbf{x} : \mathbf{v} \cdot \mathbf{x} \ge 0\}$ through the origin cut out from the points $\mathbf{a}_1, \ldots, \mathbf{a}_m$.
Such half-spaces cannot cut out all subsets of any $k + 1$ points $\mathbf{x}_1, \ldots, \mathbf{x}_{k + 1} \in \mathbb{R}^k$.
These points are linearly dependent, so there is a vector $\boldsymbol{\zeta} \ne \0$ with $\sum_j\zeta_j\mathbf{x}_j = \0$ and, say, $\zeta_j > 0$ for some $j$.
A vector $\mathbf{v}$ with $\mathbf{v} \cdot \mathbf{x}_j < 0$ for all $j$ with $\zeta_j > 0$ and $\mathbf{v} \cdot \mathbf{x}_j \ge 0$ for all other $j$ would give $0 = \sum_j\zeta_j(\mathbf{v} \cdot \mathbf{x}_j) < 0$.
So $\mathcal{U}_{\mathbf{c}, \boldsymbol{\varsigma}}$ shatters no $(k + 1)$-element subset of $[m]$, and it has at most $\sum_{i \le k}\binom mi \le (m + 1)^k$ members by \zcref{prop:sauer} with $s = k$.
In total, the families $\mathcal{U}_{\mathbf{c}, \boldsymbol{\varsigma}}$ have at most $(2m + 1)^k \cdot 2^k \cdot (m + 1)^k \le (12m^2)^k$ members.
To list the members of $\mathcal{U}_{\mathbf{c}, \boldsymbol{\varsigma}}$, call $J \subseteq [r]$ \emph{realisable} if $J = \{j \le r : \mathbf{a}_j \cdot \mathbf{v} \ge 0\}$ for some $\mathbf{v} \in O_{\boldsymbol{\varsigma}}$.
The restriction of a realisable subset of $[r + 1]$ to $[r]$ is realisable, by the same $\mathbf{v}$.
So, starting from the empty subset of $[0]$, for $r = 0, \ldots, m - 1$ we extend every realisable subset of $[r]$ in both ways and keep the realisable extensions.
A set $J \subseteq [r + 1]$ is realisable if and only if a system of linear inequalities is feasible,
 which linear programming decides~\cite{Khachiyan1979}.
The system is $\varsigma_\ell v_\ell \ge 0$ for all $\ell \in [k]$, $\mathbf{a}_j \cdot \mathbf{v} \ge 0$ for $j \in J$ and $\mathbf{a}_j \cdot \mathbf{v} \le -1$ for the other $j \le r + 1$.
Here the strict inequalities $\mathbf{a}_j \cdot \mathbf{v} < 0$ have been scaled by a positive factor, which keeps $\mathbf{v}$ in $O_{\boldsymbol{\varsigma}}$.
This takes at most $2m(m + 1)^k$ tests for every $\mathbf{c}$ and $\boldsymbol{\varsigma}$, each in time $\poly(m, k)$.
Computing the representatives and the sign vectors takes $\poly(m, k)$ operations on numbers with $\mathcal{O}(\beta)$ bits per type.
So the total time is $(m + 1)^{\mathcal{O}(k)} \cdot \poly(\beta)$.
\end{proof}

\begin{lemma}
\label{lem:exact-centerpoint}
  Let $\beta \in \N$, and let $Y \subseteq [0, 2^\beta]^k$ be a set of $m \ge 1$ points with positive integer weights.
  One can compute in time $(m + 1)^{\mathcal{O}(k)} \cdot \poly(\beta, \log|Y|)$ a point $\mathbf{c} \in (\frac12\mathbb{Z})^k \cap \bbox(Y)$ with $\dep_Y(\mathbf{c}) \ge \frac{|Y|}{k + 1}$.
\end{lemma}
\begin{proof}
  We output a deepest point of the set~$Z$ of \zcref{lem:traces}, as follows.
  We compute $Z$ and the families $\mathcal{U}_{\mathbf{c}}$ of \zcref{lem:traces} for the $m$ points of $Y$, and we give every index set the total weight of its points.
  By \zcref{lem:orthant}, the depth of a point $\mathbf{c}'$ depends only on the sign vectors $\sgn(\mathbf{y} - \mathbf{c}')$, $\mathbf{y} \in Y$, so every point of $\mathbb{R}^k$ is as deep as some point of $Z$.
  For $\mathbf{c} \in Z$, the depth $\dep_Y(\mathbf{c})$ is the least weight of a trace of an $\ell_1$-half-space through~$\mathbf{c}$, which is the least weight of a member of $\mathcal{U}_{\mathbf{c}}$, since the additional set $[m]$ has the largest weight $|Y|$.
  We output a point $\mathbf{c} \in Z$ for which this least weight is largest.
  By \zcref{lem:weighted-centerpoint}, some point of~$\mathbb{R}^k$, and hence some point of $Z$, has depth at least $\frac{|Y|}{k + 1}$, so $\dep_Y(\mathbf{c}) \ge \frac{|Y|}{k + 1} > 0$.
  Hence $\mathbf{c} \in \bbox(Y)$ by the first claim of \zcref{lem:weighted-centerpoint}.
  The families have at most $(12m^2)^k$ members, and the weight of each is computed in time $\poly(m, \log|Y|)$, so with \zcref{lem:traces}, the total time is $(m + 1)^{\mathcal{O}(k)} \cdot \poly(\beta, \log|Y|)$.
\end{proof}

\subsection{Sampling}
\label{sec:compression}
To find deep guesses, we replace the candidates by a smaller weighted set $X'$ that \emph{represents} them, that is, every $\ell_1$-half-space contains about the same fraction of $X'$ as of the candidates.
For two finite non-empty weighted sets $X', X \subset \mathbb{R}^k$, the error of this representation is
\begin{equation*}
	\err(X', X) \coloneqq \sup_{\mathbf{c} \in \mathbb{R}^k,\,\mathbf{v} \ne \0}\left|\frac{|X' \cap \halfspace{\mathbf{c}}{\mathbf{v}}|}{|X'|} - \frac{|X \cap \halfspace{\mathbf{c}}{\mathbf{v}}|}{|X|}\right|\enspace.
\end{equation*}
It satisfies the triangle inequality $\err(X'', X) \le \err(X'', X') + \err(X', X)$.
Since the supremum runs over all $\ell_1$-half-spaces, every point $\mathbf{c} \in \mathbb{R}^k$, also one computed from~$X'$, satisfies
\begin{equation}
\label{eq:transfer}
	\frac{\dep_X(\mathbf{c})}{|X|} \ge \frac{\dep_{X'}(\mathbf{c})}{|X'|} - \err(X', X)\enspace.
\end{equation}
By \zcref{lem:cutting}, the candidates form a union of few boxes, which we compress as follows.

\begin{lemma}
\label{lem:compress}
  Let $C = \Pi_1 \cup \cdots \cup \Pi_p \subset \mathbb{Z}^k$ be a non-empty disjoint union of non-empty boxes and $g \ge 1$ an integer.
  Cut every side $[a, b]$ of every box $\Pi_j$ into singletons if its length $b - a + 1$ is at most $g$, and otherwise into pieces of $\lceil(b - a + 1)/g\rceil$ consecutive integers, the last one possibly shorter.
  Give every resulting sub-box its size as weight and its smallest point as representative, and let $Y$ be the weighted set of the representatives.
  Then $Y \subseteq C$ has at most $pg^k$ points, $|Y| = |C|$, and $\err(Y, C) \le 2k/g$.
\end{lemma}
\begin{proof}
  Every side is cut into at most $g$ pieces, so every box has at most $g^k$ sub-boxes, and the sub-boxes partition $C$, so $|Y| = |C|$.
  Fix $\mathbf{c} \in \mathbb{R}^k$ and $\mathbf{v} \ne \0$.
  If, for every coordinate $i$, the $i$-th side of a sub-box is a singleton or lies strictly on one side of $c_i$, then $\sgn(\mathbf{z} - \mathbf{c})$ is the same for all its points $\mathbf{z}$.
  By \zcref{lem:orthant}, the sub-box then lies inside or outside $\halfspace{\mathbf{c}}{\mathbf{v}}$, and so does its representative.
  So the sub-box contributes the same to $|Y \cap \halfspace{\mathbf{c}}{\mathbf{v}}|$ and to $|C \cap \halfspace{\mathbf{c}}{\mathbf{v}}|$.
  Otherwise, we call the sub-box \emph{bad}, and its contributions to the two counts differ by at most its size.
  For a fixed box $\Pi_j$ and coordinate $i$, at most one piece of the $i^{\textnormal{th}}$ side has its smallest element at most~$c_i$ and its largest at least $c_i$.
  This piece has at least two elements only if the side is longer than $g$.
  In this case, the piece has at most $\lceil\ell/g\rceil < 2\ell/g$ elements, where $\ell > g$ is the length of the side.
  So the bad sub-boxes of $\Pi_j$ due to coordinate~$i$ have total size less than $\frac2g|\Pi_j|$, and summing over the $k$ coordinates and the $p$ boxes, the bad sub-boxes have total size at most $\frac{2k}g|C|$.
  Hence $|Y \cap \halfspace{\mathbf{c}}{\mathbf{v}}|$ and $|C \cap \halfspace{\mathbf{c}}{\mathbf{v}}|$ differ by at most $\frac{2k}g|C|$, and dividing by $|Y| = |C|$ gives the claim.
\end{proof}

An exact deepest point of the representatives, rounded down, is almost as deep in the candidates.
But there are up to $pg^k$ representatives, and \zcref{lem:exact-centerpoint} on $m$ points takes time~$m^{\mathcal{O}(k)}$, which is $2^{\mathcal{O}(k^2 \cdot \log(k + 1))}$ for $m = g^k$.
To reach the time bound of \zcref{thm:lfp}, we replace the representatives by only $\mathcal{O}(k^3 \cdot \log(k + 1))$ points that represent all $\ell_1$-half-spaces up to a small error.
We compute these points with a deterministic algorithm of \textcite{ChazelleM1996}, which follows an earlier scheme~\cite{Matousek1995}, and we use it only through its statement (\zcref{prop:cm}).

\begin{definition}[Set systems, $\varepsilon$-approximations and subsystem oracles~{\cite{ChazelleM1996}}]
\label{def:set-system}
A \emph{set system} $(X, \mathcal{R})$ is a finite set $X$ with a family $\mathcal{R}$ of subsets of $X$, and for $A \subseteq X$, the family $\mathcal{R}|_A$ consists of the different sets $R \cap A$ with $R \in \mathcal{R}$.
A set $A \subseteq X$ is an \emph{$\varepsilon$-approximation} for $(X, \mathcal{R})$ if $\bigl|\frac{|A \cap R|}{|A|} - \frac{|R|}{|X|}\bigr| \le \varepsilon$ for all $R \in \mathcal{R}$.
The set system has a \emph{subsystem oracle of dimension $\delta$} if $|\mathcal{R}|_A| \le \mathcal{O}(|A|)^\delta$ for all non-empty $A \subseteq X$ and if some algorithm lists, for a given non-empty $A \subseteq X$, the sets of $\mathcal{R}|_A$ in time $\mathcal{O}(|A|)^{\delta + 1}$~\cite[Definition~2.2]{ChazelleM1996}.
\end{definition}

The original statement~\cite{ChazelleM1996} writes $d$ for our $\delta$.
In \zcref{def:set-system} and in the next statement, the $\mathcal{O}$-notation hides only constants that do not depend on $\delta$, and time counts arithmetic operations~\cite[Section~1]{ChazelleM1996}.

\begin{proposition}[{\cite[Theorem~2.3]{ChazelleM1996}}]
\label{prop:cm}
Let $(X, \mathcal{R})$ be a set system with a subsystem oracle of dimension $\delta$.
For a given $r > 1$, one can compute a $(1/r)$-approximation of size $\mathcal{O}(\delta r^2 \cdot \log(\delta r))$ for $(X, \mathcal{R})$ in time $\mathcal{O}(\delta)^{3\delta} \cdot r^{2\delta} \cdot \log^\delta(\delta r) \cdot |X|$.
\end{proposition}

In words, suppose that any $m$ elements of $X$ meet the sets of $\mathcal{R}$ in only $\mathcal{O}(m)^\delta$ different ways, which can be listed quickly.
Then a subset of $X$ whose size does not depend on $|X|$ contains every set of $\mathcal{R}$ in the right proportion, up to an error $1/r$, and it is computed in time linear in $|X|$.
In our context, the elements of the ground set are the indices of a multiset $X$ of candidates, and $\mathcal{R}$ consists of the traces of the $\ell_1$-half-spaces on $X$.
Then a $(1/r)$-approximation is the index set of a sub-multiset $\hat X$ with $\err(\hat X, X) \le 1/r$.

\subsection{Finding deep guesses}
\label{sec:solver}
Since \zcref{prop:cm} counts arithmetic operations, we also need the bit sizes of the numbers that its algorithm uses.
\textcite[footnote~1]{ChazelleM1996} remark that, in the algorithms they consider, the bit size of every number is at most a constant multiple of the bit size of the input numbers. In our application, the subsystem oracle solves linear feasibility systems with at most $m+k$ inequalities in $k$ variables, with coefficients in $\{-1,0,1\}$ and right-hand sides in $\{0,-1\}$ (see the proof of \zcref{lem:traces}).
These tests have polynomial bit complexity in $m+k$, which is absorbed by choosing the constant in the oracle dimension $\delta=c_0k$ sufficiently large.

\begin{lemma}
\label{lem:deep-point}
  Let $\beta \in \N$, and let $C = \Pi_1 \cup \cdots \cup \Pi_p \subseteq [0, 2^\beta]^k$ be a non-empty disjoint union of~$p$ boxes, given by their corners.
  One can compute deterministically a point $\mathbf{w} \in \bbox(C) \cap \mathbb{Z}^k$ that is $\frac1{2(k + 1)}$-deep in $C$, in time $(2k)^{\mathcal{O}(k)} \cdot p \cdot \poly(\beta, \log p)$.
\end{lemma}
\clearpage
\pagebreak

\begin{proof}
  Let $\theta \coloneqq \frac1{8(k + 1)}$.
  We compute a multiset $\hat X \subseteq C$ of few points with \mbox{$\err(\hat X, C) \le 3\theta$}, and then an exact deepest point of $\hat X$.

  First, we compress $C$ by \zcref{lem:compress} with $g \coloneqq 16k(k + 1)$, so that $2k/g = \theta$, which gives a weighted set $Y \subseteq C$ of $m_0 \le pg^k$ points with $|Y| = |C|$ and $\err(Y, C) \le \theta$.

  Second, we round the weights to multiplicities.
  Let $K \coloneqq \lceil m_0/\theta\rceil$.
  Let $\mu_{\mathbf{y}} \coloneqq \lfloor K \cdot \omega(\mathbf{y})/|C|\rfloor$ for every $\mathbf{y} \in Y$ with weight $\omega(\mathbf{y})$, and add one to the $K - \sum_{\mathbf{y}}\mu_{\mathbf{y}}$ entries with the largest fractional parts of $K \cdot \omega(\mathbf{y})/|C|$.
  This number is the sum of these fractional parts, an integer that is at most the number of $\mathbf{y}$ with a positive fractional part.
  So only such entries are increased, and $|\mu_{\mathbf{y}} - K \cdot \omega(\mathbf{y})/|C|| < 1$ for all $\mathbf{y}$.
  Let $X$ be the multiset that contains every $\mathbf{y} \in Y$ exactly $\mu_{\mathbf{y}}$ times, so that $|X| = K$.
  For every $\ell_1$-half-space $\halfspace{\mathbf{c}}{\mathbf{v}}$, the numbers $|X \cap \halfspace{\mathbf{c}}{\mathbf{v}}|$ and $\frac{K}{|C|}|Y \cap \halfspace{\mathbf{c}}{\mathbf{v}}|$ differ by at most $\sum_{\mathbf{y}}|\mu_{\mathbf{y}} - K \cdot \omega(\mathbf{y})/|C|| < m_0$, and dividing by $K = |X|$ gives $\err(X, Y) \le m_0/K \le \theta$.

  Third, we apply \zcref{prop:cm} to the set system $([K], \mathcal{R})$, where $[K]$ indexes the points of $X$ and $\mathcal{R}$ consists of the traces of the $\ell_1$-half-spaces on $X$.
  For non-empty $A \subseteq [K]$, the sets of $\mathcal{R}|_A$ are the traces on the sub-multiset indexed by $A$.
  By \zcref{lem:traces}, there are at most $(12|A|^2)^k \le (4|A|)^{2k}$ of them, and the union of the families of \zcref{lem:traces} for the points indexed by $A$, with repetitions removed, lists exactly them.
  The listing takes $(|A| + 1)^{\mathcal{O}(k)}$ arithmetic operations, in which the coordinates enter only through comparisons and midpoints.
  Hence, for a suitable absolute constant $c_0 \ge 2$, the system has a subsystem oracle of dimension $\delta \coloneqq c_0k$.
  \zcref{prop:cm} with $r \coloneqq 1/\theta$ gives a sub-multiset $\hat X \subseteq X$, non-empty since $[K] \in \mathcal{R}$, of $\mathcal{O}(\delta r^2 \cdot \log(\delta r)) = \mathcal{O}(k^3 \cdot \log(k + 1))$ points with $\err(\hat X, X) \le \theta$.
  By the triangle inequality, $\err(\hat X, C) \le 3\theta$.

  Finally, \zcref{lem:exact-centerpoint}, applied to the distinct points of $\hat X$ with their multiplicities as weights, gives $\mathbf{c} \in \bbox(\hat X)$ with $\dep_{\hat X}(\mathbf{c}) \ge \frac{|\hat X|}{k + 1}$, and we output $\mathbf{w} \coloneqq \lfloor\mathbf{c}\rfloor$.
  By~\eqref{eq:transfer}, $\dep_C(\mathbf{c}) \ge \bigl(\frac1{k + 1} - \frac3{8(k + 1)}\bigr)|C| \ge \frac{|C|}{2(k + 1)}$, and by \zcref{lem:rounding}, $\mathbf{w}$ is at least as deep.
  Since $\hat X \subseteq X \subseteq Y \subseteq C$ as sets of points, $\mathbf{c} \in \bbox(C)$, and so $\mathbf{w} \in \bbox(C)$ by \zcref{lem:rounding}.

  \subparagraph*{Analysis of the run time.}
  We have $m_0 \le pg^k = p \cdot (2k)^{\mathcal{O}(k)}$ and $K \le m_0/\theta + 1 = p \cdot (2k)^{\mathcal{O}(k)}$.
  The compression and the rounding take $\mathcal{O}(m_0 \cdot (k + \log m_0))$ operations on numbers with $\mathcal{O}(k\beta + \log K)$ bits.
  The algorithm of \zcref{prop:cm} takes $\mathcal{O}(\delta)^{3\delta} \cdot r^{2\delta} \cdot \log^\delta(\delta r) \cdot K = (2k)^{\mathcal{O}(k)} \cdot p$ operations.
  Its input numbers are the indices in $[K]$ and the coordinates that our subsystem oracle uses, so, by the remark at the beginning of this section, all its numbers have $\mathcal{O}(\beta + \log K)$ bits.
  Any bit size polynomial in $K$ and $\beta$ would also suffice for \zcref{thm:lfp}, since the lemma would then take time $(2k)^{\mathcal{O}(k)} \cdot \poly(p, \beta)$, and the proof of \zcref{thm:lfp} applies it with $p = (2k \cdot \log(N + 2))^{\mathcal{O}(k)}$ and $\beta = \mathcal{O}(k \cdot \log(N + 2))$.
  The linear programs of the subsystem oracle have coefficients in $\{-1, 0, 1\}$ and are solved in time $\poly(|A|, k)$, which is absorbed by the choice of $c_0$.
  \zcref{lem:exact-centerpoint} on $\mathcal{O}(k^3 \cdot \log(k + 1))$ points takes time $(2k)^{\mathcal{O}(k)} \cdot \poly(\beta, \log K)$.
  So the total time is $(2k)^{\mathcal{O}(k)} \cdot p \cdot \poly(\beta, \log p)$.
\end{proof}

\begin{proof}[\bf Proof of \zcref{thm:lfp}]
  Let $\mathrm{DCC}_k$, for \emph{deterministic centerpoint cutting}, be the following algorithm.
  If $N = 0$, then $[0, N]^k = \{\0\}$, and it returns the least fixed point~$\0$ of $F$ without evaluations.
  Otherwise, it runs the cutting algorithm of \zcref{sec:cutting} with the deterministic $\frac1{2(k + 1)}$-deep guesses of \zcref{lem:deep-point}, applying the lemma to $C$ written as the union of its boxes (\zcref{lem:cutting}).
  So, by \zcref{lem:cutting} with $\alpha = \frac1{2(k + 1)}$, the search stops after $\mathcal{O}(k^2 \cdot \log(kN + 2))$ rounds and $\mathcal{O}(k^3 \cdot \log(kN + 2))$ evaluations of $F$, and it outputs $\mathbf{q}^*$.

  For the overhead, the time besides the evaluations, let $\vartheta \coloneqq 2k \cdot \log(N + 2) \ge 2$, so that $k \le \vartheta$ and $\log(kN + 2) \le \log k + \log(N + 2) \le \vartheta$.
  The number of rounds is $\mathcal{O}(k^2 \cdot \log(kN + 2)) = \vartheta^{\mathcal{O}(1)}$, so in every round $C$ is the union of $p = \vartheta^{\mathcal{O}(k)}$ boxes.
  The bookkeeping of a round takes $\vartheta^{\mathcal{O}(k)}$ operations on integers with $\mathcal{O}(\vartheta)$ bits, by \zcref{lem:cutting}.
  All coordinates are at most $\lambda N = 16k^2N^3$, so \zcref{lem:deep-point} applies with an integer $\beta = \mathcal{O}(\vartheta)$ and takes time $(2k)^{\mathcal{O}(k)} \cdot p \cdot \poly(\beta, \log p) = \vartheta^{\mathcal{O}(k)}$.
  Forming $F'(\mathbf{w})$, $\surplus'(\mathbf{w})$ and $\rd(\mathbf{w})$ from the at most $k + 1$ values of $F$, as in the proof of \zcref{lem:sinks}, takes time $\poly(k, \vartheta)$ per round.
  Hence the overhead is $\vartheta^{\mathcal{O}(1)} \cdot \vartheta^{\mathcal{O}(k)} = \vartheta^{\mathcal{O}(k)}$.
\end{proof}

\section{Finishing the proof}%
\label{sec:proof}

\begin{proof}[\bf Proof of \zcref{thm:main}]

We use \zcref{algo:blocks} with the set $S$ computed by \zcref{lem:pivots} in polynomial time and partitioned into blocks.
The block solver, given $k' \in [k]$, runs the algorithm of \zcref{thm:lfp} with $k'$ in place of $k$.
The input has $t^+ \le 2^n$ tokens, so \zcref{lem:recursion} applies.
The algorithm returns the numbers $t^-_v$ of the input instance, and it remains to choose $d$ and $k$ and to bound the run time.

Let $L \coloneqq n^{1/3} \cdot \log^{2/3} n$, let $d \coloneqq \lceil L\rceil$ and $k \coloneqq \lceil L/\log n\rceil$.
Since $n \ge 2$ (\zcref{sec:instances}), we have $1 \le \log n \le n$, so $0 < L \le n$ and $1 \le k \le d \le n$.

\medskip
We now analyse the run time.
The hypothesis of \zcref{lem:recursion} concerns boxes $[0, N]^{k'}$ with $k' \in [k]$ and $N \le 2^{4n^2}$, so that $\log(N + 2) \le 4n^2 + 1$.
The bounds of \zcref{thm:lfp} grow with $k'$ and $N$.
So on these boxes the block solver makes at most $Q = \mathcal{O}(k^3 \cdot n^2) = \mathcal{O}(n^5)$ evaluations, and its overhead is at most $I = (2k \cdot (4n^2 + 1))^{\mathcal{O}(k)}$, which gives
\begin{equation*}
	\log(Q + 1) = \mathcal{O}(\log n)\qquad\text{and}\qquad\log I = \mathcal{O}(k \cdot \log n)\enspace.
\end{equation*}
We may assume $I \ge 1$, so that $2^d + I \le 2 \cdot 2^d \cdot I$.
By \zcref{lem:recursion}, the logarithm of the run time is then at most
\begin{equation*}
	M \cdot \log(Q + 1) + d + \log I + \mathcal{O}(\log n)\enspace.
\end{equation*}
We bound the three terms with the parameters.
Since $k \ge L/\log n$ and $d \ge L$, we have $kd \ge L^2/\log n = n^{2/3} \cdot \log^{1/3} n$, and hence
\begin{equation*}
	M \cdot \log(Q + 1) = \mathcal{O}\Bigl(\Bigl(\frac n{kd} + 1\Bigr) \cdot \log n\Bigr) = \mathcal{O}\bigl(n^{1/3} \cdot \log^{2/3} n + \log n\bigr) = \mathcal{O}(L)\enspace.
\end{equation*}
Next, $d \le L + 1$.
Finally, $k \le L/\log n + 1$ gives $\log I = \mathcal{O}(k \cdot \log n) = \mathcal{O}(L + \log n) = \mathcal{O}(L)$.
Since $\log n \le L$, the logarithm of the run time is $\mathcal{O}(L)$, that is, the run time is $2^{\mathcal{O}(n^{1/3} \cdot \log^{2/3} n)}$.

For \arrival, the reduction of \zcref{sec:instances} gives a \garrival instance on $n + 1$ vertices with $t^+ = 1 \le 2^{n + 1}$.
The bound becomes $2^{\mathcal{O}(n^{1/3} \cdot \log^{2/3} n)}$, since $L(n + 1) \le 2L(n)$ for $n \ge 2$, where $L(m) \coloneqq m^{1/3} \cdot \log^{2/3} m$.
\end{proof}

\printbibliography

\end{document}